\documentclass[11pt]{article}
\usepackage[a4paper,left=25mm,right=25mm,top=25mm, bottom=25mm]{geometry}
\usepackage[usenames]{color}
\usepackage{a4wide}

\usepackage{amsmath}
\usepackage{amssymb}
\usepackage{amsthm}
\usepackage{amsbsy}
\usepackage{tikz}

\usepackage{natbib}
\usepackage{hyperref}

\usepackage{epsfig}
\usepackage{amsfonts}
\usepackage{color}
\newtheorem{lemma}{Lemma}
\newtheorem{theorem}{Theorem}
\newtheorem{definition}{Definition}
\newtheorem{proposition}{Proposition}

\newtheorem{corollary}{Corollary}
\newtheorem{claim}{Claim}  

\newtheorem{example}{Example}
\renewenvironment{proof}{\par\noindent\underline{Proof}: }{\hfill$\square$\par\vskip10pt}

\renewcommand{\baselinestretch}{1.3}\normalsize
\begin{document}

\title{Strategy-proofness with single-peaked and single-dipped preferences}

\author{Jorge Alcalde--Unzu\thanks{\footnotesize{Department of Economics and INARBE, Universidad P\'ublica de Navarra, Campus Arrosadia, 31006
Pamplona, Spain. \texttt{Email:\,jorge.alcalde@unavarra.es}. Financial support from
the Spanish Ministry of Economy and Competitiveness, through project PID2021-127119NB-I00 (funded by MCIN/AEI/10.13039/501100011033 and by ``ERDF A way of making Europe'') is gratefully acknowledged.}}, \,
Oihane Gallo\thanks{\footnotesize{Corresponding author. Faculty of Business and Economics, University of Lausanne, Internef, 1015, Lausanne,
Switzerland. \texttt{Email:\,oihane.gallo@unil.ch}. Main financial support from the Swiss National Science Foundation (SNSF) through project 100018$\_$192583, and additional support from the Spanish Ministry of Economy and Competitiveness through projects PID2021-127119NB-I00 and PID2019-107539GB-I00 (funded by MCIN/AEI/10.13039/501100011033 and by ``ERDF A way of making Europe'') is gratefully acknowledged.}}\, and Marc
Vorsatz\thanks{\footnotesize{Departamento de An\'alisis Econ\'omico,
Universidad Nacional de Educaci\'on a Distancia (UNED), Paseo Senda del Rey 11, 28040
Madrid, Spain. \texttt{Email:\,mvorsatz@cee.uned.es}. Financial support from
the Spanish Ministry of Economy and Competitiveness, through project PID2021-122919NB-I00 (funded by MCIN/AEI/10.13039/501100011033 and by ``ERDF A way of making Europe'') is gratefully acknowledged.}} }
\thispagestyle{empty}
\date{March 2024}
\maketitle
\begin{abstract}\vspace{0.1cm}
\noindent We analyze the problem of locating a public facility in a domain of single-peaked and single-dipped preferences when the social planner knows the type of preference (single-peaked or single-dipped) of each agent. 
Our main result characterizes all strategy-proof rules and shows that they can be decomposed into two steps. 
In the first step, the agents with single-peaked preferences are asked about their peaks and, for each profile of reported peaks, at most two alternatives are preselected. 
In the second step, the agents with single-dipped preferences are asked to reveal their dips to complete the decision between the preselected alternatives. 
Our result generalizes the findings of \cite{moulin1980strategy} and \cite{barbera1994characterization} for single-peaked and of \cite{manjunath2014efficient} for single-dipped preferences.
Finally, we show that all strategy-proof rules are also group strategy-proof and analyze the implications of Pareto efficiency. \vspace{0.15cm}
\end{abstract}
\vspace{0.15cm}

\noindent \textit{JEL-Numbers:} D70, D71, D79.

\newpage
\renewcommand{\baselinestretch}{1.5}\normalsize
\section{Introduction} 

\subsubsection*{Motivation and main results}

\noindent Governments continually improve their cities by constructing new public facilities such as schools, hospitals, or parks. When deciding upon the location of these facilities, public officials consider technical constraints (\emph{e.g.}, not all locations may be feasible) and monetary limitations (\emph{e.g.}, the construction costs could differ from one location to another) but they may also take the preferences of the affected population into account. However, since preferences are private information and agents are strategic, it cannot be ruled out that agents misrepresent their preferences, which could adversely affect the final decision. 
This paper seeks to construct social choice rules that always incentivize agents to reveal their preferences truthfully, a property known as strategy-proofness. 

\medskip

\noindent A central result in the literature on strategy-proofness is the impossibility of \cite{gibbard1973manipulation} and \cite{satterthwaite1975strategy}, which states that any strategy-proof rule on the universal preference domain with more than two alternatives in its range is dictatorial. Therefore, to construct non-dictatorial social choice rules that induce truth-telling, one has to restrict either the range of the rules to two alternatives or the domain of admissible preferences. Since rules with a range of two alternatives are not Pareto efficient on the universal preference domain, the literature has focused on identifying situations in which the preference domain can be naturally restricted. 

\medskip

\noindent As we explain in more detail in our literature review, the domain of single-peaked and the domain of single-dipped preferences have received special attention when decision makers decide where to construct a new public facility. In our general approach, we combine these two preference domains to address the problem of locating a public facility in  any subset of the real line. To be more specific, in our model some agents have single-peaked preferences, while others have single-dipped preferences. Furthermore, the type of preference of each agent (single-peaked or single-dipped) is commonly known but the location of the peak/dip together with the rest of the preference is private information. In this setting, the set of admissible preferences for an agent with single-peaked (single-dipped) preferences is equal to the set of all single-peaked (single-dipped) preferences.

\medskip

\noindent Our main result (Theorem 1) characterizes all strategy-proof rules on the aforementioned domain. In particular, we find that all strategy-proof rules comply with the following two-step procedure. In the first step, each agent with single-peaked preferences is asked about her best alternative in the range of the rule (her ``peak"), and for each profile of reported peaks, at most two alternatives are preselected. If only one alternative is preselected, this is the final outcome. If two alternatives are preselected, these two alternatives necessarily form a pair of contiguous alternatives in the range of the rule. In the second step, each agent with single-dipped preferences is asked about her worst alternative in the range of the rule (her ``dip"). Then, taking into account the information of all ``peaks" and ``dips", one of the two preselected alternatives is chosen. 

\medskip

\noindent We also show that all strategy-proof rules are group strategy-proof (Theorem 1). Finally, the range of any strategy-proof and Pareto efficient rule is either equal to the set of alternatives or coincides with the ``extreme points" of the set of alternatives (Proposition 5).

\subsubsection*{Related literature}

There is a considerable literature on strategy-proofness that studies domains in which the Gibbard-Satterthwaite impossibility can be avoided; see \cite{barbera2011strategyproof} for a survey. For the case of locating public facilities, the domains of single-peaked preferences and single-dipped preferences have been analyzed in depth.

\medskip

\noindent If the facility is a public good, \emph{e.g.}, a school or a hospital, the domain of single-peaked preferences is appealing. \cite{black1948decisions, black1948rationale} was the first to discuss single-peaked preferences and to show that the median voter rule, which selects the median of the declared peaks, is strategy-proof and selects the Condorcet winner. Later, \cite{moulin1980strategy} and \cite{barbera1994characterization} characterize all strategy-proof rules on this domain: the {\it generalized median voter rules}. By contrast, if the facility is considered a public bad, \emph{e.g.} a prison, the single-dipped preference domain emerges naturally. \cite{barbera2012domains} and \cite{manjunath2014efficient} show that on this domain all strategy-proof rules have a range of at most two.

\medskip

\noindent However, there are public facilities which do not give rise to unanimous opinions, that is, the preferences of some agents are single-peaked, while the preferences of others are single-dipped.
Examples include dog parks, soccer stadiums, or dance clubs. For example, dog owners may have single-peaked preferences on the location of a dog park, yet individuals who do not like dogs probably have single-dipped preferences. Consequently, in order to deal with such situations, a domain that includes both single-peaked and single-dipped preferences is needed. If the set of admissible preferences of each agent coincides with the set of all single-peaked and all single-dipped preferences, \cite{berga2000maximal} and \cite{achuthankutty2018dictatorship} show that the Gibbard-Satterthwaite impossibility applies. Hence, the preference domain needs to be further constrained. 

\medskip

\noindent The closest works to our analysis are \cite{alcalde2018strategy}, \cite{thomson2022should} and \cite{feigenbaum2015strategyproof}. \cite{alcalde2018strategy} characterize all strategy-proof rules when the peak/dip of each agent, which coincides with the location of the agent on the real line, is public information but the social planner neither knows the type of preference of an agent (single-peaked or single-dipped) nor the rest of the agent's preference. Unlike them, we consider a domain in which the type of preference of each agent (single-peaked or single-dipped) is public information but the social planner neither knows the location of the peak/dip nor the rest of the preference of an agent. Observe that in our setting, in contrast to \cite{alcalde2018strategy}, agents could misrepresent their preferences by lying about the location of the peak/dip, but not about their type of preference (\emph{i.e.}, an agent with single-peaked preferences cannot declare that she has single-dipped preferences, and vice versa).

\medskip

\noindent \cite{thomson2022should} studies a restricted domain with two agents. In this domain, the social planner has the following information: (i) the first agent has single-peaked and the second agent has single-dipped preferences,  and (ii) the peak of the preference of the first agent and the dip of the preference of the second agent are situated at the same publicly known location. \cite{thomson2022should} shows that in this domain, the only Pareto efficient and strategy-proof rules are dictatorial. Observe that \cite{thomson2022should} differs from our model in two aspects. First, we do not restrict the number of agents and allow for any distribution of agents with single-peaked and single-dipped preferences. Second, and most importantly, in \cite{thomson2022should} the social planner has more information about the agents' preferences. In fact, \cite{thomson2022should} assumes that the social planner knows the locations of the peak and dip (and even that they coincide), while this information is not available to the social planner in our domain. And, as we will detail later, this difference is crucial for our possibility results. 

\medskip

\noindent Finally, \cite{feigenbaum2015strategyproof} introduce a domain where the type of preference of each agent (single-peaked or single-dipped) is public information, while the location of the peaks and dips of the preferences are unknown, as it happens in our domain. Unlike us, they assume that the preferences over the alternatives outside the peak or dip are cardinally determined by the distance to the peak or dip. Furthermore, the objective of both papers also differs as \cite{feigenbaum2015strategyproof} seek strategy-proof rules that best approximate certain social values like utilitarianism or Rawlsianism.

\medskip


\noindent  Our characterization establishes that the first step of the strategy-proof rules on our domain is similar to the strategy-proof rules on the single-peaked preference domain, and  the second step is similar to the strategy-proof rules on the single-dipped preference domain. As a consequence, previous results in the literature are generalized. To see this, remember that in the first step of the two-step procedure, only the agents with single-peaked preferences are asked about their peaks and that either a single alternative or a pair of contiguous alternatives is preselected.  Proposition 3 shows that if we define an intuitive order on the possible sets of preselected alternatives, a generalized median voter function has to be applied in the first step. Moreover, in the second step of the two-step procedure, a binary decision problem is faced and we find that the choice between the two preselected locations is made in the same way as in the strategy-proof rules on the single-dipped preference domain (Proposition 4). Thus, if the set of agents with single-dipped preferences is empty, only the first step applies and only single alternatives appear as the outcome of the first step. Consequently, the two-step procedure reduces to the standard generalized median voter rules of \cite{moulin1980strategy} and \cite{barbera1994characterization}. Similarly, if the set of agents with single-peaked preferences is empty, either a unique singleton or a unique pair of alternatives is preselected. If a unique singleton is preselected, the rule is constant. If a unique pair of alternatives is preselected, the two-step procedure reduces to the rules characterized in \cite{manjunath2014efficient}. We will be more specific about these relations after presenting our main characterization in Section 3.

\subsubsection*{Remainder}

\noindent The remainder of the paper is organized as follows. Section \ref{sec1} presents the model. Our results are in Section \ref{sec2}. Section \ref{sec4} concludes. All proofs are relegated to the Appendix.

\section{The model}\label{sec1}

Let $X \subseteq \mathbb{R}$ be a set of feasible alternatives and let $N=\{1,\ldots,n\}$ be a finite set of agents that is partitioned into two groups: the set of agents $A$ with cardinality $|A|=a \in \mathbb{N} \cup \{0\}$ and the set of agents $D=N\setminus A$ with cardinality $|D|=n-a \in \mathbb{N} \cup \{0\}$. 
Let $R_i$ be the preference of agent $i \in N$ over $X$. 
Formally, $R_i$ is a complete, transitive, and antisymmetric binary relation. 
The strict preference relation induced by $R_i$ is denoted by $P_i$.
Let ${\cal R}_{i}$ be the preference domain of agent $i$. 
We refer to ${\cal R} = \times_{i \in N} {\cal R}_i$ as the domain.

\medskip

\noindent The agents belonging to $A$ have single-peaked preferences. 
That is, each $R_i \in {\cal R}_i$ of an agent $i \in A$ satisfies the following condition: there is an alternative $\rho(R_i) \in X$ (called the \emph{peak}) such that for each $x, y \in X$ with $\rho(R_i) \geq x > y$ or $\rho(R_i) \leq x < y$, it follows that $x \, P_i \, y$. 
Furthermore, each agent $i \in D$ has single-dipped preferences.
That is, each $R_i \in {\cal R}_i$ of an agent $i \in D$ satisfies the following condition: there is an alternative $\delta(R_i) \in X$ (called the \emph{dip}) such that for each $x, y \in X$ with $\delta(R_i) \geq x > y$ or $\delta(R_i) \leq x < y$, it follows that $y \, P_i \, x$. 
If $i\in A$, then ${\cal R}_{i}$ is the set of all single-peaked preferences over $X$.
Similarly, if $i \in D$, then ${\cal R}_{i}$ is the set of all single-dipped preferences over $X$.

\medskip

\noindent A preference profile is a list of preferences $R \equiv (R_i)_{i \in N} \in {\cal R}$. 
For each $S \subset N$, let ${\cal R}^S = ({\cal R}_i)_{i \in S}$ be the subdomain of ${\cal R}$ restricted to the agent set $S$.  
Given profile $R \in {\cal R}$, $R_S \in {\cal R}^S$ and $R_{-S} \in {\cal R}^{N \setminus S}$ are obtained by restricting $R$ to $S$ and $N\setminus S$, respectively.
We can thus write $R=(R_A,R_D)$.
Given $R=(R_A,R_D) \in {\cal R}$, $\rho(R)$ denotes the vector of peaks of the agent set $A$ at $R$.
Similarly, $\delta(R)$ is the vector of dips of the agent set $D$ at $R$.

\medskip

\noindent A social choice rule is a function $f: {\cal R} \rightarrow X$ that selects for each preference profile $R \in {\cal R}$ an alternative $f(R) \in X$. 
Let $\Omega = \{x \in X \, : \, \exists R \in {\cal R} \mbox{ such that } f(R) = x\}$ be the range of $f$. Then, $\Omega(\boldsymbol{\rho})= \{x \in \Omega : \exists R \in {\cal R} \mbox{ such that } \rho(R) = \boldsymbol{\rho} \mbox{ and } f(R) = x\}$ is the range of $f$ for the subdomain of preferences when the agent set $A$ has the vector of peaks $\boldsymbol{\rho}$.
For each $i \in A$ and each $R_i \in {\cal R}_i$, let $p(R_i) = \{x \in \Omega \, : \, x \, P_i \, y \mbox{ for each } y \in \Omega \setminus \{x\}\}$ be the $\Omega$-restricted peak of $R_i$.
Similarly, for each $i \in D$ and each $R_i \in {\cal R}_i$, let $d(R_i) = \{x \in \Omega \, : \, y \, P_i \, x \mbox{ for each } y \in \Omega \setminus \{x\}\}$ be the $\Omega$-restricted dip of $R_i$. Observe that given any value of $\rho(R_i)$ for any $i \in A$ (respectively, $\delta(R_i)$ for any $i \in D$), there are at most two possible values for $p(R_i)$ (respectively, $d(R_i)$). We define the vector of $\Omega$-restricted peaks at $R \in{\cal R}$ as $p(R)\equiv(p(R_i))_{i \in A}$.
Similarly, $d(R) \equiv (d(R_i))_{i \in D}$ is the vector of $\Omega$-restricted dips at $R \in{\cal R}$. To ensure the existence of $p(R)$ and $d(R)$, we assume that $\mathbb{R} \setminus \Omega$ is either empty or the union of open sets.\footnote{In fact, if this condition is not satisfied, it is not only that $p(R)$ and $d(R)$ cannot be defined, but also that there are instances in which no strategy-proof rule exists.} 
Denote the range of $f$ for the subdomain of preferences when the agent set $A$ has the $\Omega$-restricted peaks $\textbf p \in \Omega^a$ by $\Omega(\textbf p)$.
That is, $\Omega(\textbf p) = \{x \in \Omega : \exists R \in {\cal R} \mbox{ such that } p(R) = \textbf p \mbox{ and } f(R) = x\}.$
Finally, let $\Omega^2 = \{(x, y) \in X^2 \, : \, x, y \in \Omega \mbox{ and } x \neq y\}$ be the set of pairs formed by alternatives of $\Omega$ and denote by $\Omega^{2}_{C}=\{(x,y) \in \Omega^2 : x < y  \mbox{ and } (x, y) \cap \Omega=\emptyset\}$ the set of all ordered pairs of contiguous alternatives of $\Omega$.\footnote{Observe that $\Omega_C^2$ is empty if either ($i$) $|\Omega| = 1$ and, thus, $f$ is constant, or ($ii$) $\Omega$ is an interval and, thus, there are no contiguous alternatives.}$^{,}$\footnote{Several concepts introduced in the paragraph depend on the actual rule $f$. Since we do not consider variations of $f$ throughout the manuscript, this dependence is not made explicit in the notation.} 

\medskip

\noindent We focus on rules that are incentive compatible. 
A rule $f$ is said to be manipulable by group $S \subseteq N$ if there is a preference profile $R\in {\cal R}$ and a subprofile $R'_S\in{\cal R}^{S}$ such that for each $i \in S$, $f(R'_S, R_{-S}) \, P_i \, f(R)$.
Then, a rule $f$ is Group Strategy-Proof (GSP) if it is not manipulable by any group $S \subseteq N$. 
Similarly, $f$ is said to be Strategy-Proof (SP) if it is not manipulable by any group $S \subseteq N$ of size $|S| = 1$. 
A rule $f$ is Pareto Efficient (PE) if for each $R\in{\cal R}$, there is no $x\in X$ such that $x \, P_i \, f(R)$ for each $i\in N$. 

\section{Strategy-proof rules}\label{sec2}

\noindent In this section, we provide a characterization of all strategy-proof rules on our domain. 
Since the social choice rule $f$ is necessarily SP whenever $|\Omega| = 1$, we focus only on rules with $|\Omega| \geq 2$.
Our first result is essential and establishes that on the subdomain of preferences induced by any vector of $\Omega$-restricted peaks $\textbf p \in \Omega^a$, the size of the range is at most 2.

\begin{proposition}
\label{barbera}
Let $f:\mathcal{R}\rightarrow\Omega$ be SP. Then, for each vector of $\Omega$-restricted peaks $\emph{\textbf p} \in \Omega^a$, $|\Omega(\emph{\textbf p})| \leq 2$.
\end{proposition}

\noindent Proposition \ref{barbera} implies that any SP rule can be decomposed into two steps. 
In the first step, the agents with single-peaked preferences indicate their $\Omega$-restricted peaks.
This first step is thus independent of the agents with single-dipped preferences.
The outcome of this first step consists of a set of at most two preselected alternatives. Formally, this first step is a function $\omega : \Omega^a \rightarrow \Omega \cup \Omega^2$ that assigns to each vector of $\Omega$-restricted peaks $\textbf p$ a set $\omega(\textbf p) = \Omega(\textbf p)$ of size $|\Omega(\textbf p)| \leq 2$.\footnote{If the set of agents with single-peaked preferences is empty, \emph{i.e.}, $a=0$, we similarly define $\omega(\emptyset) \in \Omega\cup\Omega^{2}$.}
We define $\underline{\omega}(\textbf{p}) = \min \Omega(\textbf{p})$ and $\overline{\omega}(\textbf{p}) = \max\Omega(\textbf{p})$. 
We can observe at this point that if $\omega(\textbf p)$ is a singleton and $\underline \omega(\textbf p) = \overline\omega(\textbf p)$, then alternative $\omega(\textbf p)$ is chosen by $f$ for all $R \in {\cal R}$ such that $p(R)=\textbf p$.
If two alternatives are preselected and $\underline{\omega}(\textbf{p}) \neq \overline{\omega}(\textbf{p})$, then a binary decision function $g_{\omega(\textbf p)} : {\cal R} \rightarrow \omega(\textbf p)$ is applied to each profile $R \in {\cal R}$ with $p(R)=\textbf p$. Observe that the second-step binary decision function may depend on the outcome $\omega(\textbf p)$ of the first step. Corollary \ref{structure0} summarizes this two-step procedure.\footnote{\cite{alcalde2023structure} show that SP rules on any domain can be decomposed into a two-step procedure.}

\begin{corollary}
\label{structure0}
Let $f: {\cal R} \rightarrow \Omega$ be SP. Then, there is a function $\omega : \Omega^a \rightarrow \Omega \cup \Omega^2$ and a set of binary decision functions $\{g_{\omega(\emph{\textbf p})} : {\cal R} \rightarrow \omega(\emph{\textbf p})\}_{\omega(\emph{\textbf p}) \in \Omega^2}$ such that for each $\emph{\textbf p} \in \Omega^a$ and each $R \in {\cal R}$ with $p(R)= \emph{\textbf p}$, $$f(R) = \left\{
	\begin{array}{ll}
	\omega({\bf p})  & \mbox{ if } \omega({\bf p}) \in \Omega \\*[5pt]
	
	g_{\omega({\bf p})}(R) & \mbox{ if } \omega({\bf p}) \in \Omega^2.
	\end{array}
	\right.$$
\end{corollary}

\noindent In the remainder of this section we analyze the additional conditions SP imposes on $\omega$ and on the set of binary decision functions $\{g_{\omega(\emph{\textbf p})}\}_{\omega(\emph{\textbf p}) \in \Omega^2}$.

\subsection*{The structure of $\omega$}

\noindent We first discuss the conditions SP imposes on the range of $\omega$.
Denote the range of $\omega$ by $r_{\omega} \equiv \{\alpha \in \Omega \cup \Omega^2 \, : \, \exists {\bf p} \in \Omega^a \mbox{ such that } \omega({\bf p}) = \alpha\}$. A first implication is that if $\omega(\textbf p) \in \Omega^2$, then the two preselected alternatives are contiguous in $\Omega$ and, thus, $\omega(\textbf p)$ is an element of $\Omega^{2}_{C}$. This implies that $r_\omega$ is a subset of $\Omega\cup\Omega^{2}_{C}$. A second implication is that for all elements in the interior of $\Omega$, \emph{i.e.}, for all $x \in \mathbf{int}(\Omega) \equiv \Omega \setminus \{\min \Omega, \max \Omega\}$, there is a vector of $\Omega$-restricted peaks $\textbf p$ such that $x$ is the unique preselected alternative. Finally, we have an implication concerning the extreme elements of $\Omega$ (when they exist). As they belong to $\Omega$, each of them should be preselected by $\omega$ either alone or with its contiguous alternative for some vector of $\Omega$-restricted peaks.\footnote{Observe that if there is $x \in \mathbf{int}(\Omega)$ such that $[\min \Omega, x] \subseteq \Omega$ (respectively, $[x, \max \Omega] \subseteq \Omega$), then there exists a vector of $\Omega$-restricted peaks such that $\min \Omega$ (respectively, $\max \Omega$) is the unique preselected alternative. Otherwise, the first implication would be contradicted.} This is formally stated in the following proposition.

\begin{proposition}
\label{rangeomega}
\label{omega}
Let $f: {\cal R} \rightarrow \Omega$ be SP. Then, $r_{\omega}$ is such that:
\begin{itemize}
    \item[$(i)$] $\emph{\textbf{int}}(\Omega) \subseteq r_{\omega} \subseteq \Omega\cup \Omega^{2}_{C}$.
    \item[$(ii)$] If $\min \Omega$ exists and there is no $x \in \Omega$ such that $(\min \Omega, x) \in r_{\omega}$, then $\min \Omega \in r_{\omega}$.
    \item[$(iii)$] If $\max \Omega$ exists and there is no $x \in \Omega$ such that $(x, \max \Omega) \in r_{\omega}$, then $\max \Omega \in r_{\omega}$.
\end{itemize}
\end{proposition}

\noindent The following example illustrates the implications of Proposition \ref{rangeomega}.

 \begin{example}
\label{ex1}
Let $X = \mathbb{R}$. If $f$ is SP and has range $\Omega = [1, 2] \cup \{3, 4\}$, then the range $r_{\omega}$ of the first-step function $\omega$ must be a subset of $\Omega \cup \Omega^2_C = [1, 2] \cup \{3, 4\} \cup \{(2,3), (3,4)\}$. Moreover, the interior of $\Omega$, $\textbf{int}(\Omega) = (1, 2] \cup \{3\}$, must belong to $r_{\omega}$. Observe that there is freedom to include the elements of $\Omega^2_C$ into $r_{\omega}$, but there are some limitations in the case of  $(3, 4)$: if $4\notin r_{\omega}$, then $(3, 4)\in r_{\omega}$, as $4 \in \Omega$. This always occurs with the elements of $\Omega^2_C$ that involve an extreme alternative of $\Omega$. Note that if there is no element of $\Omega^2_C$ that involves one of these extremes (because that extreme does not have a contiguous alternative in $\Omega$), as it occurs here with $1$, then this extreme alternative must belong to $r_{\omega}$ given that this is the unique possibility to include it in $\Omega$. Therefore, we have six possibilities for $r_{\omega}$ in this example: $[1, 2] \cup \{(2, 3), 3, 4\}$, $[1, 2] \cup \{(2, 3), 3, (3, 4)\}$, $[1, 2] \cup \{(2, 3), 3, (3, 4), 4\}$, $[1, 2] \cup \{3, 4\}$, $[1, 2] \cup \{3, (3, 4)\}$ and $[1, 2] \cup \{3, (3, 4), 4\}$.
\end{example}

\noindent We establish in Proposition~\ref{generalized} that $\omega$ is a generalized median voter rule. For that, we first define the complete and transitive binary relation $\leq^*$ over $\Omega\cup \Omega^{2}_{C}$, where $<^*$ represents its asymmetric part, as follows:
\begin{itemize}
\item[(i)] for each $x, y \in \Omega$, $x \leq^* y \Leftrightarrow x \leq y$,
\item[(ii)] for each $(x,y) \in \Omega^{2}_{C}$, $x <^* (x,y) <^* y$.
\end{itemize}

\noindent It can be seen that $\leq^*$ orders all real numbers of $\Omega$ in the standard way. 
And, as Figure $1$ shows for the case when $\Omega = [1, 2] \cup \{3, 4\}$, each pair of contiguous alternatives $(x,y) \in \Omega^{2}_{C}$ is inserted between $x$ and $y$.

\medskip

\begin{figure}[ht]
	\centering
	\begin{tikzpicture}[scale=0.85]
	\tikzset{lw/.style = {line width=1pt}}
	\draw [-, line width=0.45mm](0,0) -- (6,0);
	
	
	\foreach \x in {0}
	\draw[-, line width=0.45mm] (\x cm,3pt) -- (\x cm,-3pt);
	
	
	\foreach \x in {2}
	\draw[-, line width=0.45mm] (\x cm,3pt) -- (\x cm,-3pt);
	
	\foreach \x in {3}
	\draw[-] (\x cm,1.5pt) -- (\x cm,-1.5pt);
	
	\foreach \x in {4}
	\draw[-, line width=0.45mm](\x cm,3pt) -- (\x cm,-3pt);
	
	\foreach \x in {5}
	\draw [-](\x cm,1.5pt) -- (\x cm,-1.5pt);
	
	\foreach \x in {6}
	\draw[-, line width=0.45mm] (\x cm,3pt) -- (\x cm,-3pt);

	\draw (0,0) node[below=3pt] {$\textbf{1}  $} node[above=3pt] {$   $};
	\draw (2,0) node[below=3pt] {$\textbf{2}  $} node[above=3pt] {$   $};
	\draw (3,0) node[below=3pt] {$\textbf{(2,3)}  $} node[above=3pt] {$   $};
	\draw (4,0) node[below=3pt] {$\textbf{3}  $} node[above=3pt] {$   $};
	\draw (5,0) node[below=3pt] {$\textbf{(3,4)}  $} node[above=3pt] {$   $};
	\draw (6,0) node[below=3pt] {$\textbf{4}  $} node[above=3pt] {$   $};

	\end{tikzpicture}
	\caption{ Order $<^{*}$ over $\Omega\cup \Omega^{2}_{C} = [1,2] \cup \{(2,3), 3, (3,4), 4\}$.}
\end{figure}
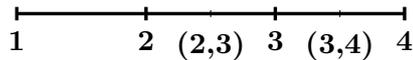

\medskip

\noindent In order to apply $\omega$ we only require a set of agents with single-peaked preferences and a set $r_{\omega} \subseteq \Omega\cup \Omega^{2}_{C}$ of elements with an order $\leq^*$. It is well-known in the literature \citep[see][]{barbera2011strategyproof} that the SP rules on a domain in which all agents have single-peaked preferences over an ordered set are generalized median voter rules. However, this result cannot be applied here because even though the preferences of the agents of $A$ are single-peaked, the domain of the sub-society $A$ over $\Omega\cup \Omega^{2}_{C}$ cannot be described by the cartesian product of individual preference domains.\footnote{As a consequence, the domain belongs to the category of \emph{interdependent values environments} \citep[see, for instance,][]{barbera2022restricted}.} To see this, consider an agent whose peak is at $x \in \Omega$. Then, it seems that this agent is going to strictly prefer $(x, y)$ over $y$ as the outcome of the function $\omega$. However, this is not always the case because the outcome depends on the entire preference profile. For example, the preference domain of agent $i$ only includes the preferences for which $(x, y) \, I_i \, y$ and not those for which $(x, y) \, P_i \, y$ whenever the actual preferences of the remaining agents imply that, if $x$ and $y$ are preselected in the first step, $y$ is finally selected in the second step. Nevertheless, it turns out that the unique $\omega$ functions that are compatible with a SP rule on our domain are also generalized median voter functions.\footnote{ \cite{rodriguez2017single} also analyzes the SP rules on a context where agents have single-peaked preferences and the outcome of the rule can be either a single alternative or a pair of adjacent alternatives. However, we cannot apply his result directly because the author assumes that preferences are also single-peaked over $\Omega\cup \Omega^{2}_{C}$.} To define this family of rules, we need to introduce first the concept of a left coalition system.

\begin{definition}
\label{leftsystem}
\emph{Given a social choice rule} $f: {\cal R} \rightarrow \Omega$, \emph{with} $r_{\omega}$ \emph{satisfying the conditions of Proposition \ref{omega}}, \emph{a} left coalition system on $r_{\omega}$ \emph{is a correspondence} $\mathcal{L} : r_{\omega} \rightarrow 2^A$ \emph{that assigns to each} $\alpha \in r_{\omega}$ \emph{a collection of coalitions} $\mathcal{L}(\alpha)$ \emph{such that:}
\begin{enumerate}
\item[$(i)$] \emph{if} $C\in \mathcal{L}(\alpha)$ \emph{and} $C \subset C'$, \emph{then} $C' \in \mathcal{L}(\alpha)$,
\item[$(ii)$] \emph{if} $\alpha <^* \beta$ \emph{and} $C \in \mathcal{L}(\alpha)$, \emph{then} $C \in \mathcal{L}(\beta)$, \emph{and}
\item[$(iii)$] \emph{if} $\max\Omega$ \emph{exists and} $\max \Omega \notin r_{\omega}$, \emph{then} $\emptyset\in {\cal L}(\max r_{\omega}) \setminus{\cal L}(\alpha)$ \emph{for each} $\alpha\in r_{\omega}\setminus\{\max r_{\omega}\}$.
\item[$(iv)$] \emph{if} $\max\Omega$ \emph{does not exist, then} $\emptyset\notin {\cal L}(\alpha)$ \emph{for each} $\alpha\in r_{\omega}$.
\end{enumerate}
\end{definition}

\noindent A left coalition system on $r_{\omega}$ includes for each element $\alpha \in r_{\omega}$ a set of coalitions ${\cal L}(\alpha)$. The coalitions in ${\cal L}(\alpha)$ can be interpreted as ``support" or ``winning coalitions" that are needed so that an alternative to the left of or equal to $\alpha$ under the order $\leq^{*}$ is implemented. Condition $(i)$ states that if a coalition is winning at $\alpha$, all its supercoalitions are also winning at $\alpha$. Condition $(ii)$ states that if a coalition is winning at $\alpha$, it is also winning at any $\beta>^{*}\alpha$. Finally, as we will see later, conditions $(iii)$ and $(iv)$ guarantee that any alternative in $\Omega$ is effectively an outcome of $f$.

\medskip

 \noindent We make use of Example \ref{ex1} to illustrate the concept of a left coalition system.

\medskip

\noindent \textbf{Continuation of Example \ref{ex1}}: \emph{Suppose that $r_{\omega} = [1, 2] \cup \{(2, 3) , 3, (3, 4)\}$ and let $A = \{i_1, i_2, i_3\}$. An example of a left coalition system is as follows: ${\cal L}(x) = \{S \subseteq A \, : \, |S| \geq 2\}$ for each $x \in [1, 2]$, ${\cal L}(2, 3) = {\cal L}(3) = \{S \subseteq A \, : \, |S| \geq 1\}$ and ${\cal L}(3,4) = 2^A$. Observe that for each element in $r_{\omega}$, the winning coalitions are those that exceed a minimal size. Moreover, this minimal size is weakly decreasing in $\leq^{*}$. Therefore, conditions $(i)$ and $(ii)$ are obviously satisfied. Similarly, since $\max \Omega = 4 \not\in r_{\omega}$, it is required by $(iii)$ that $\emptyset \in {\cal L}(3,4) \setminus {\cal L}(3)$, which is also satisfied in this construction.}

\medskip

\noindent We formally define a generalized median voter function.

\begin{definition}
\label{voterscheme}
\emph{Given a social choice rule $f: \mathcal {R}\rightarrow \Omega$ and a left coalition system $\mathcal{L}$ on $r_{\omega}$, with $r_{\omega}$ satisfying the conditions of Proposition \ref{omega}, 
the} associated generalized median voter function $\mathbf{\omega}$ \emph{is as follows: for each ${\bf p} \in \Omega^a$ and each $R \in {\cal R}$ such that} $p(R) = {\bf p}$,
$$\omega(\mathbf{p})= \alpha \mbox{\emph{ if and only if}}$$ $$ \{i \in A: p(R_i) \leq^* \alpha\} \in \mathcal{L}(\alpha)$$ \
\emph{and} 
$$\{i \in A: p(R_i) \leq^* \beta\}\notin \mathcal{L}(\beta)\, \mbox{ \emph{for each} } \beta \in r_{\omega} \mbox{ \emph{such that} } \beta <^* \alpha.$$
\end{definition}

\noindent The generalized median voter function $\omega$ associated with $f$ and the left coalition ${\cal L}$ on $r_{\omega}$ chooses the first element $\alpha\in r_{\omega}$, starting from the left, such that the set of agents whose $\Omega$-restricted peaks are to the left of or at $\alpha$ belongs to ${\cal L}(\alpha)$.\footnote{See \cite{jennings2023new} for alternative terminology and definitions of the strategy-proof rules in the single-peaked context.}

\medskip

 \noindent \textbf{Continuation of Example \ref{ex1}}: \emph{We discuss the outcome of $\omega$ at some subprofiles for the left coalition system ${\cal L}$ of the example:}

\begin{itemize}
\item \emph{Consider a subprofile $R_A \in {\cal R}^A$ such that $p(R_A) =(p(R_{i_{1}}), p(R_{i_{2}}), p(R_{i_{3}}))=(1, 1.5, 4)$. Since $\{i \in A \, : \, p (R_i) \leq^{*} 1.5\} = \{i_1, i_2\} \in {\cal L}(1.5)$ and $\{i \in A \, : \, p(R_i) \leq^{*} x\} = \{i_1\} \not\in {\cal L}(x)$ for each $x <^* 1.5$, we have that $\omega(p(R_A)) = 1.5 \in \Omega$. Then, $f(R) = 1.5$.}
\item \emph{Consider a subprofile $R'_A \in {\cal R}^A$ such that $p(R'_A) =(p(R'_{i_{1}}), p(R'_{i_{2}}), p(R'_{i_{3}}))= (1, 3, 4)$. Since $\{i \in A \, : \, p(R'_i) \leq^{*} (2, 3)\} = \{i_1\} \in {\cal L}(2, 3)$ and $\{i \in A \, : \, p(R'_i) \leq^{*} 2\} = \{i_1\} \not\in {\cal L}(2)$, we have that $\omega(p(R'_A)) = (2, 3)$.}
\item \emph{Consider a subprofile $R''_A \in {\cal R}^A$ such that $p(R''_A) =(p(R''_{i_{1}}), p(R''_{i_{2}}), p(R''_{i_{3}}))=(4, 4, 4)$. Since $\{i \in A \, : \, p (R''_i) \leq^{*} (3, 4)\} = \emptyset \in {\cal L}(3, 4)$ and $ \{i \in A \, : \, p(R''_i) \leq^{*} 3\} = \emptyset \notin {\cal L}(3)$, we have that $\omega(p(R''_A)) = (3, 4)$.}
\end{itemize}


\medskip

\noindent We are now ready to explain the relevance of conditions $(iii)$ and $(iv)$ of Definition \ref{leftsystem}. 
First, observe that if $\emptyset \in{\cal L}(\alpha)$ for some $\alpha\in r_{\omega}$, there is no alternative to the right of $\alpha$ that is selected by $\omega$. The underlying reasoning is as follows. If $\emptyset\in{\cal L}(\alpha)$ for some $\alpha\in r_{\omega}$, then, by condition $(i)$ of Definition \ref{leftsystem}, ${\cal L}(\alpha)=2^A$. Therefore, for each $R\in{\cal R}$, $\{i \in A: p(R_i) \leq^* \alpha\} \in \mathcal{L}(\alpha)$.
This implies that $\omega(p(R))\leq^{*}\alpha$. Consequently, if $\emptyset\in{\cal L}(\alpha)$ for some $\alpha\in r_{\omega}$, then $\max r_{\omega}$ exists, which implies that $\max\Omega$ exists as well. Thus, condition $(iv)$ of Definition \ref{leftsystem} requires that if $\max\Omega$ does not exist, then $\emptyset \not \in {\cal L}(\alpha)$ for any $\alpha \in r_{\omega}$. To understand condition $(iii)$, suppose that $\Omega = [1, 2] \cup \{3,4\}$ and that $4 \not\in r_{\omega}$. On the one hand, if $\emptyset\in{\cal L}(\alpha)$ for some $\alpha\neq(3,4)$, then by condition $(i)$ of Definition \ref{leftsystem}, ${\cal L}(\alpha)=2^A$. Therefore, for each $R\in{\cal R}$, $\omega(p(R))\leq^{*} \alpha$ and $(3,4)$ will never appear as the outcome of $\omega$. Hence, alternative $4$ will never be the outcome of $f$, which is not possible given that $4\in \Omega$. On the other hand, consider subprofile $R''_A$ in Example \ref{ex1}. If $\emptyset\notin {\cal L}(\alpha)$ for each $\alpha\in r_{\omega}$, then $\omega(p(R''_A))>^{*}(3,4) = \max r_{\omega}$, which is not possible.


\medskip

\noindent Finally, the following proposition introduces the structure of the first step of each SP rule on our domain.

\begin{proposition}
\label{generalized}
Let $f: {\cal R} \rightarrow \Omega$ be SP. Then, the function $\omega$ is a generalized median voter function on a set $r_{\omega}$, with $r_{\omega}$ satisfying the conditions of Proposition \ref{omega}. 
\end{proposition}

\subsection*{The structure of the functions $g_{\omega(\mathbf{p})}$}

\noindent We analyze the structure SP imposes on the second step; \emph{i.e.}, on the collection of binary decision rules $\{g_{\omega(\mathbf{p})} : {\cal R} \rightarrow \omega(\mathbf{p})\}_{\omega(\mathbf{p}) \in \Omega^2}$. 
As we have already indicated before, if $\omega(\mathbf{p}) \in \Omega^2$, then the two preselected alternatives are contiguous. Therefore, we can concentrate on the binary decision rules $\{g_{\omega(\mathbf{p})}\}_{\omega(\mathbf{p}) \in r_{\omega}\cap\Omega^2_{C}}$. 

\medskip

\noindent We define a particular class of binary decision functions called voting by collections of left-decisive sets. 
Each of these functions $g_{\omega(\mathbf{p})}$ is defined by specifying a set of coalitions $W(g_{\omega(\mathbf{p})}) \subseteq 2^N$ (called left-decisive sets) such that $g_{\omega(\mathbf{p})}$ chooses $\underline{\omega}(\mathbf{p})$ if and only if the set of agents that prefer $\underline{\omega}(\mathbf{p})$ to $\overline{\omega}(\mathbf{p})$ is a superset of any coalition that belongs to $W(g_{\omega(\mathbf{p})})$. To introduce the formal definition, we need the following concepts.
For each $\omega(\mathbf{p}) \in r_{\omega}\cap\Omega^2_{C}$ and each $R \in {\cal R}$ such that $p(R)=\mathbf{p}$, let $L_{\omega(\mathbf{p})}(R)$ be the set of agents that prefer $\underline{\omega}(\mathbf{p})$ to $\overline{\omega}(\mathbf{p})$ at $R$.
We say that a coalition $S\subseteq N$ is minimal in a set of coalitions $V\subseteq 2^N$ if $S\in V$ and for each $S'\subset S, S'\notin V$. 
Then, a set of coalitions $V\subseteq 2^N$ is minimal if all its coalitions are minimal in that set.

\begin{definition}
\label{winningdef}
 \emph{Consider a social choice rule} $f: \mathcal {R}\rightarrow \Omega$ \emph{and a generalized median voter function} $\omega$ \emph{on a set} $r_{\omega}$, \emph{with} $r_{\omega}$ \emph{satisfying the conditions of Proposition \ref{omega}}, \emph{associated with a left coalition system} ${\cal L}$. \emph{Then, for each} ${\bf p} \in \Omega^a$ \emph{such that} $\omega(\mathbf{p}) \in r_{\omega} \cap \Omega^2_{C}$, \emph{the associated binary decision function} $g_{\omega({\bf p})}$ \emph{is a} voting by collections of left-decisive sets \emph{if there is a minimal set of coalitions} $W(g_{\omega({\bf p})}) \subseteq 2^N$ \emph{such that for each} $R \in {\cal R}$ \emph{with} $p(R) = {\bf p}$,
	$$g_{\omega({\bf p})}(R) = \left\{
	\begin{array}{ll}
	\underline{\omega}(\mathbf{p}) & \mbox{ \emph{if} } C \subseteq L_{\omega({\bf p})}(R) \mbox{ \emph{for some} } C \in W(g_{\omega({\bf p})}) \\*[5pt]
	\overline{\omega}(\mathbf{p}) & \mbox{ \emph{otherwise}, } 
	\end{array}
	\right. \vspace{0.2cm}$$
	\emph{and the following conditions are satisfied:}
	\begin{itemize}
		\item[(i)] \emph{For each} $C\in W(g_{\omega({\bf p})})$, $C\cap D\neq \emptyset$.
		\item[(ii)] \emph{For each minimal coalition} $B$ of ${\cal L}(\omega(\mathbf{p}))\setminus{\cal L}(\underline{\omega}(\mathbf{p}))$, \emph{there is} $C\in W(g_{\omega(\mathbf{p})})$ \emph{such that} $C\cap A=B$.
		\end{itemize}
\end{definition}

\noindent We illustrate Definition~\ref{winningdef} with help of Example \ref{ex1}.

\medskip

\noindent \textbf{Continuation of Example \ref{ex1}}: \emph{Suppose that $D = \{j_1, j_2, j_3\}$ and let a voting by collections of left-decisive sets be such that $W(g_{(2,3)}) = W(g_{(3, 4)}) = \{S \subseteq N \, : \, |S| = 3\, \mbox{ and }\, |S\cap A|=1\}$. We discuss the outcome of $g_{(2, 3)}$ and $g_{(3, 4)}$ at some profiles:}
\begin{itemize}
\item \emph{Consider a subprofile $R'_D \in {\cal R}^D$ such that $d(R') = (d(R'_{j_1}), d(R'_{j_2}), d(R'_{j_3})) = (1, 3, 3)$. To determine $f(R') = f(R'_A, R'_D)$, remember that $\omega(p(R')) = (2,3)$ and $p(R'_A) =(p(R'_{i_{1}}), p(R'_{i_{2}}), p(R'_{i_{3}}))= (1, 3, 4)$. Then, we analyze $g_{(2,3)}$.} Observe that $L_{(2, 3)}(R') = \{i_1, j_2, j_3\}$. Given that $|\{i_1, j_2, j_3\}|=3$ and  $|\{i_1, j_2, j_3\}\cap A|=1$, we have $L_{(2, 3)}(R')\in W(g_{(2, 3)})$. Thus, $g_{(2,3)}(R') = 2$ and $f(R') = 2$.
\item \emph{Consider a subprofile $R''_D \in {\cal R}^D$ such that $d(R'') = (d(R''_{j_1}), d(R''_{j_2}), d(R''_{j_3})) =(1, 2, 4)$. To determine $f(R'') = f(R''_A, R''_D)$, remember that $\omega(p(R'')) = (3, 4)$ and $p(R''_A) =(p(R''_{i_{1}}), p(R''_{i_{2}}), p(R''_{i_{3}}))= (4, 4, 4)$. Then, we analyze $g_{(3,4)}$.} Observe that $L_{(3, 4)}(R'') = \{j_3\}$. Given that $|\{j_3\}|=1$, we have that there is no $C \in W(g_{(3, 4)})$ such that $C \subseteq L_{(3, 4)}(R'')$. Thus, $g_{(3,4)}(R'') = 4$ and $f(R'') = 4$.
\end{itemize}


\medskip

\noindent Observe that two additional conditions are required to complete the description of the minimal set of coalitions $W(g_{\omega(\mathbf{p})})$. These conditions guarantee that, given any $\mathbf{p}\in \Omega^{a}$ such that $\omega(\mathbf{p})\in r_{\omega}\cap\Omega^2_{C}$, for each alternative $\underline{\omega}(\mathbf{p})$ and $\overline{\omega}(\mathbf{p})$, there is a profile with vector of $\Omega$-restricted peaks $\mathbf{p}$ such that the alternative is chosen by $f$. To see this, observe that for each $\mathbf{p}\in \Omega^{a}$, with $\omega(\mathbf{p})\in r_{\omega}\cap\Omega^2_{C}$, we have that the set of agents of $A$ that prefer $\underline{\omega}(\mathbf{p})$ to $\overline{\omega}(\mathbf{p})$, say $B$, is in ${\cal L}(\omega(\mathbf{p}))\setminus{\cal L}(\underline{\omega}(\mathbf{p}))$. Then, condition $(ii)$ requires that there is a coalition of $W(g_{\omega(\mathbf{p})})$ such that its agents with single-peaked preferences are exactly those agents of $B$. Additionally, note that given any $\mathbf{p}\in \Omega^{a}$ such that $\omega(\mathbf{p})\in r_{\omega}\cap\Omega^2_{C}$, the preference between $\underline{\omega}(\mathbf{p})$ and $\overline{\omega}(\mathbf{p})$ of all agents of $A$ is known. Then, if the preferences of the agents with single-dipped preferences are not considered, the outcome of $f$ will be either always $\underline{\omega}(\mathbf{p})$ or always $\overline{\omega}(\mathbf{p})$ for all profiles with vector of $\Omega$-restricted peaks $\mathbf{p}$, which is not possible because $\omega(\mathbf{p})\in r_{\omega}\cap\Omega^2_{C}$. Therefore, condition $(i)$ requires that in each coalition of $W(g_{\omega(\mathbf{p})})$ there has to be at least one agent with single-dipped preferences.

\medskip

\noindent The next proposition indicates the structure of the second step of each SP rule on our domain.

\begin{proposition}
\label{second-step2}
Let $f: {\cal R} \rightarrow \Omega$ be SP and $\omega$ be its associated generalized median voter function on $r_\omega$, with $r_{\omega}$ satisfying the conditions of Proposition \ref{omega}. Then, the family of binary decision functions $\{g_{\omega(\mathbf{p})} :  {\cal R} \rightarrow \omega(\mathbf{p})\}_{\omega(\mathbf{p})\in r_{\omega}\cap\Omega^2_{C}}$ is such that for each $\omega(\mathbf{p}) \in r_{\omega}\cap\Omega^2_{C}$, $g_{\omega(\mathbf{p})}$ is a voting by collections of left-decisive sets.
\end{proposition}

\subsection*{The characterization}\label{result}

Our main result establishes that the necessary conditions derived from SP in the former propositions are also sufficient. Moreover, since our domain ${\cal R}$ satisfies the condition of \emph{indirect sequential inclusion} introduced in \cite{barbera2010individual}, there is an equivalence between SP and GSP. Therefore, all characterized rules are also GSP.

\begin{theorem}
\label{theorem}
The following statements are equivalent:
\begin{itemize}
\item[(i)] $f:\mathcal{R}\rightarrow\Omega$ is SP.
\item[(ii)] $f:\mathcal{R}\rightarrow\Omega$ is GSP.
\item[(iii)] There is a generalized median voter function $\omega$ on a set $r_\omega$, with $r_{\omega}$ satisfying the conditions of Proposition \ref{omega}, and a set of voting by collections of left-decisive sets $\{g_{\omega(\mathbf{p})} :  {\cal R} \rightarrow \omega(\mathbf{p})\}_{\omega(\mathbf{p})\in r_{\omega}\cap\Omega^2_{C}}$ such that for each $R\in {\cal R}$ with $p(R)=\mathbf{p}$, $$f(R) = \left\{
	\begin{array}{ll}
	\omega({\bf p})  & \mbox{ if } \omega({\bf p}) \in \Omega \\*[5pt]
	
	g_{\omega({\bf p})}(R) & \mbox{ if } \omega({\bf p}) \in \Omega^{2}_{C}.
	\end{array}
	\right.$$
\end{itemize}
\end{theorem}

\medskip 

\noindent We discuss several important aspects of Theorem \ref{theorem}. First and foremost, Theorem \ref{theorem} generalizes \cite{moulin1980strategy} and \cite{barbera1994characterization} (who consider the case when all agents have single-peaked preferences) and \cite{manjunath2014efficient} (who assumes that all agents have single-dipped preferences).
Their results emerge from our theorem in the following way.
\begin{itemize}
\item Proposition \ref{second-step2} establishes that if $\omega({\bf p})  \in \Omega_C^2$, then $g_{\omega(\bf p)}$ is a voting by collection of left-decisive sets. 
By part ($i$) of Definition 3, the winning coalitions of $W(g_{\omega(\bf p)})$ must contain some agent with single-dipped preferences.
So, if all agents have single-peaked preferences, there is no voting by collections of left-decisive sets that can be applied in the second step.
Therefore, the range of $\omega$ is $\Omega$ and Proposition \ref{generalized} dictates that $f$ is a generalized median voter rule on $\Omega$ \cite[see][]{moulin1980strategy, barbera1994characterization}.

\item If all agents have single-dipped preferences, the $\Omega$-restricted vector of peaks ${\bf p}$ is not determined ---{\it i.e}, ${\bf p}={\bf \emptyset}$--- 
and, by footnote 2, $\omega({\bf \emptyset}) \in \Omega \cup \Omega^2$. 
Hence, if $\omega({\bf \emptyset}) \in \Omega$, $f$ is a constant rule. And if $\omega({\bf \emptyset})\in \Omega^2$, $f$ is a voting by collections of left-decisive sets between the two preselected alternatives \cite[see][]{manjunath2014efficient}.
\end{itemize}

\noindent Second, it is worth highlighting that not much information about the agents' preferences is needed in order to determine the outcome of the characterized rules: for each preference profile $R \in \mathcal{R}$, any SP rule $f$ only requires the vectors of $\Omega$-restricted peaks and dips, $p(R)$ and $d(R)$, as input.
This is because $(i)$ $\omega$ only takes the vector of $\Omega$-restricted peaks $p(R)$ into account, and $(ii)$ $g_{\omega({p(R)})}$ only requires the agents' preferences between the preselected alternatives $\underline{\omega}({p(R)})$ and $\overline{\omega}({p(R)})$. Since $\Omega$ does not contain any alternative strictly between $\underline{\omega}({p(R)})$ and $\overline{\omega}({p(R)})$, the vectors of $\Omega$-restricted peaks and dips, $p(R)$ and $d(R)$, provide the necessary information. 

\medskip

\noindent Third, and perhaps surprisingly, the characterized rules show a strong asymmetry between the agents with single-peaked and those with single-dipped preferences. 
In fact, we have seen that only the agents with single-peaked preferences matter in the first step.
The preferences of the agents with single-dipped preferences are only considered in the second step when a decision between two contiguous alternatives has to be taken. This implies that the agents with single-dipped preferences are unable to provoke huge outcome swings on their own. Moreover, if the range of the rule is an interval, the agents with single-dipped preferences are totally irrelevant to determine the outcome: since there are no contiguous alternatives in the range of the rule, it never happens that two alternatives pass to the second step.

\medskip

\noindent Fourth, the social planner can design any of the characterized rules by applying the following procedure. 
In the beginning, the possible outcomes of the rule are determined; \emph{i.e.}, the social planner chooses a range $\Omega \subseteq X$ for the rule. Next, the social planner determines $r_{\omega}$ by deciding both which pairs of $\Omega^{2}_{C}$ are included in $r_{\omega}$ and whether or not $\min \Omega$ and/or $\max \Omega$, if they exist, belong to $r_{\omega}$ with the restrictions imposed in Proposition \ref{omega}. 
Afterwards, a left coalition system ${\cal L}$ on $r_{\omega}$ is defined in such a way that $\omega$ is a generalized median voter function associated with it. Finally, since the outcome of $\omega$ may belong to $\Omega^{2}_{C}$, the social planner specifies for each $\omega(\mathbf{p})\in r_{\omega}\cap\Omega^2_{C}$, a set of minimal left-decisive sets $W(g_{\omega(\mathbf{p})})$. This procedure can be used to define many rules beyond those characterized in \cite{moulin1980strategy}, \cite{barbera1994characterization} and \cite{manjunath2014efficient}. One example is the rule that we have constructed step by step in Example \ref{ex1} throughout this section. That rule takes into account the preferences of the agents of both $A$ and $D$.

\subsubsection*{Strategy-proofness and Pareto efficiency}\label{sec3}

Since there are many SP rules on our setting, we impose PE as an additional axiom in order to tighten the family characterized in Theorem \ref{theorem}. The results are summarized in the following proposition.

\begin{proposition}\label{PE}
The following statements hold:
\begin{itemize}
\item[$(i)$] All SP rules $f$ such that $\Omega = X$ are PE.
\item[$(ii)$] All SP rules $f$ such that $\Omega \not\in \{X, \{\min X, \max X\}\}$ are not PE.
\item[$(iii)$] All SP rules $f$ such that $\Omega =\{\min X, \max X\}\neq X$ are PE if and only if $A=\emptyset$ or $r_\omega = \{(\min X, \max X)\}$.
\end{itemize}
\end{proposition}

\noindent It can be seen that PE imposes a strong restriction on the range of $f$: only a range equal to the set of feasible alternatives or a range equal to the extreme alternatives of that set is compatible with SP and PE. In particular, any SP rule whose range is equal to the set of feasible alternatives is PE, while only a subset of those SP rules with range equal to the extremes of the set of feasible alternatives is PE. Note that if the set of feasible alternatives does not have a minimum or maximum, then PE can be only achieved when the range of $f$ is the entire set of feasible alternatives. 

\medskip

\noindent The intuition behind Proposition \ref{PE} is as follows. Note first that if a SP rule $f$ has a range equal to $X$ but is not PE, then there is an alternative in $X$ that is unanimously preferred to the one chosen by $f$. Since that alternative is in the range of $f$, the rule is not GSP and, by Theorem \ref{theorem}, nor is SP. Observe that this argument only explains part $(i)$ in Proposition \ref{PE}. Now we focus on the remaining cases.

\medskip

\noindent First, it is easy to check that if $\min X$ and $\max X$ exist, then they must be in the range of $f$ to guarantee PE. To see why, assume, for instance, that $\min X$ exists, but it is not in the range of $f$ and consider a preference profile in which all the peaks are in $\min X$ and all the dips are in $\max X$. Then, $\min X$ Pareto dominates any alternative in the range of $f$, so $f$ is not PE.

\medskip

\noindent Second, if $\min X$ and $\max X$ exist, but they are not the only elements in the range of $f$, then there is an interior alternative in $\Omega$. Moreover, since the range of $f$ does not coincide with $X$, there is another interior alternative in $X\setminus \Omega$. Let $x\in X\setminus \Omega$ and $y\in \Omega$, with $y$ being the closest alternative to $x$. Consider a preference profile where all the peaks are located in $x$ and all the $\Omega$-restricted peaks and dips are in $y$. It can be seen that any SP rule chooses $y$ in that profile by Theorem \ref{theorem}. However, $x$ Pareto dominates $y$, which contradicts PE. Therefore, if a SP rule $f$, with $\Omega \neq X$, is PE, then $\Omega=\{\min X,\max X\}$.

\medskip

\noindent Finally, if $\min X$ and $\max X$ exist and the range of $f$ is equal to $\{\min X,\max X\}$, then more restrictions are needed to guarantee PE. These restrictions are provided in part $(iii)$ of Proposition \ref{PE}. Observe that these conditions require the set $D$ to be non-empty. Moreover, by the role of the agents of $D$ in $W(g_{(\min X,\max X)})$ (see condition ($i$) in Definition 3), $r_{\omega} = \{(\min X, \max X)\}$ implies that if $f$ chooses $\min X$ (respectively, $\max X$) it is because there is an agent of $D$ that prefers $\min X$ (respectively, $\max X$) to any other alternative of $X$, which guarantees PE.

\section{Concluding remarks}\label{sec4}

We analyze the problem of locating a public facility that is considered by some agents in the society a good and by others a bad. Since the Gibbard-Satterthwaite impossibility applies if the set of admissible preferences of each agent includes all single-peaked and all single-dipped preferences, preferences have to be further restricted. We propose a new domain according to which the type of preference of each agent (single-peaked or single-dipped) is known but it is private information as to where each agent's peak/dip is located and how each agent ranks the rest of the alternatives. This domain fits well with situations in which, even though the location for each agent is publicly known, that location may not necessarily coincide with her peak/dip. For instance, if the public facility is a nursery, parents with children may consider this facility desirable, but it might be undesirable for others without children or for those who prefer to live in a quiet neighborhood. Moreover, despite the fact that home addresses are registered, people spend a considerable amount of time at work and some parents may prefer to have a nursery close to their workplace rather than to their home. Note that such situations cannot be accommodated in the domain of \cite{alcalde2018strategy}, who assume that the peak/dip of each agent coincides with the publicly known location of the agent. In contrast, to allow each agent total flexibility as to the location of her peak/dip, the social planner needs to have full information about the type of preference of each agent, while in the domain of \cite{alcalde2018strategy} that information is private. 

\medskip

\noindent We characterize all strategy-proof rules on this new domain. It is relevant to point out that the characterized family generalizes existing results in the literature; in particular, the results in \cite{moulin1980strategy} and  \cite{barbera1994characterization} for only single-peaked and the result in \cite{manjunath2014efficient} for only single-dipped preferences. We also find out that the characterized family can be easily implemented in two steps and with few information. Finally, we establish that all strategy-proof rules are also group strategy-proof and show that Pareto efficiency implies a strong restriction on the range of the strategy-proof rules.

\medskip

\noindent The family of strategy-proof rules characterized here shows some similarities with and differences from the family of strategy-proof rules characterized in \cite{alcalde2018strategy}. The strategy-proof rules on the domain of \cite{alcalde2018strategy} also follow a two-step procedure. In their model, the location of the peak/dip of each agent is known, so the first step of their rules asks which agents have single-peaked preferences. As a result of the first step, both the type of preference of each agent and the location of the peaks and dips are known. In the domain analyzed here, the type of preference of each agent is public information and in the first step we ask agents with single-peaked preferences about their peaks. As a result, the type of preference of each agent and the location of all peaks are known. Note that even though the social planner in our domain has less information after the first step  (since she does not know the location of the dips), at most two alternatives are preselected in both settings. If two alternatives are preselected, the second step of \cite{alcalde2018strategy} asks all agents to order the preselected alternatives. Unlike them, we only need to ask agents with single-dipped preferences about their dips. Thus, less information is required for the strategy-proof rules on our domain to be implemented: we only need to know the location of the peaks and dips, while they need to know the type of preferences of each agent and the ordinal preference between any pair of preselected alternatives that may appear as the outcome of the first step. Moreover, we provide a closed-form characterization of the strategy-proof rules on our domain, while \cite{alcalde2018strategy} do not. 

\medskip

\noindent We next comment on the possibility to extend our characterization result to some relevant subdomains, that is, if the preference domains of the agents are further restricted by assuming that the set of admissible preferences for the agents of $A$ are not all single-peaked preferences and/or  the set of admissible preferences for the agents of $D$ are not all single-dipped preferences. Example \ref{ex2} below shows that even if only very minor additional restrictions are imposed, new strategy-proof rules can emerge. 

\begin{example}
\label{ex2}
Let $N = \{i, j\}$ and $X = \{1, 3, 4\}$. Assume that the preference domain of agent $i$ consists of all single-peaked preferences over $X$ and that the preference domain of agent $j$ consists of all single-dipped preferences over $X$ but preference $4 P_j 1 P_j 3$. Consider a rule $f$, with $\Omega = X$, such that for each $R \in {\cal R}$, 
$$f(R) = \left\{
	\begin{array}{ll}
	
	3  & \mbox{ if } 3 P_i 4 \mbox{ and } 4 P_j 3 P_j 1 \\*[5pt]
	4  & \mbox{ if } 4 P_i 3 \mbox{ and } 4 P_j 3 P_j 1\\*[5pt]
 1  & \mbox{ otherwise. } 
	\end{array}
	\right.$$
 
\noindent This rule is SP, but it does not belong to the family characterized in Theorem \ref{theorem}. To see this, note that $\omega(p(R))$ of this rule takes the following values: $\omega(1) = \omega(3) = (1, 3)$ and $\omega(4) = (1, 4)$. Since pair $(1,4)$ does not contain contiguous alternatives, condition $(i)$ of Proposition \ref{omega} is violated.
\end{example}

\noindent In Example \ref{ex2} preferences over alternatives for agent $j$ are determined by the distance to the dip, where each feasible alternative can potentially be the dip. We can construct a similar subdomain for agent $i$ that includes all single-peaked preferences with peak in one of the feasible alternatives and such that preferences are defined by the distance to the peak. This subdomain is obtained by eliminating the preference $3 \, P_i \, 1 \, P_i \, 4$ from the preference domain of agent $i$ in Example \ref{ex2}. Note that the rule constructed in Example \ref{ex2} does not belong to the characterized family in this subdomain either (again, $\omega(4)$ is the non-contiguous pair $(1,4)$), but it is strategy-proof. Consequently, it follows that if preferences are determined by the distance to the peak/dip, our characterization result does not necessarily hold if peaks and dips are assumed to be in feasible locations.

\medskip


\noindent Another interesting subdomain is due to \cite{thomson2022should} who studies a domain for the two-agent case, one agent $i$ with single-peaked and another agent $j$ with single-dipped preferences. This domain, unlike ours, assumes that the location of the peak of $i$ and the location of the dip of $j$ coincide in a publicly known location. \cite{thomson2022should} establishes then that under SP and PE, all rules are dictatorial. However, Example \ref{ex3} below highlights that the assumption about the location of the peak and dip is crucial because in our domain there are non-dictatorial SP and PE rules with more than two alternatives in the range even if $|A| = |D| = 1$.

\begin{example}\label{ex3}
Let $X= \{1, 2, 3\}$, $A = \{i\}$, and $D = \{j\}$. Consider the rule $f$, with range $\Omega = X$, such that for each $R \in {\cal R}$:
$$f(R) = \left\{
	\begin{array}{ll}
	p(R_i)  & \mbox{ if } p(R_i) \in \{1, 2\} \\*[5pt]
	
	2 & \mbox{ if } p(R_i) = d(R_j) = 3 \\*[5pt]

 3 & \mbox{ otherwise. }
	\end{array}
	\right.$$

\medskip

\noindent This rule is non-dictatorial, SP, PE and has three alternatives in the range. The information the social planner has about the peaks and dips in the domain of  \cite{thomson2022should} prevents the construction of such rule therein.   
\end{example}

\medskip

\noindent Considerations for further research include restrictions of the family characterized by imposing additional axioms such as anonymity or extending the model by allowing for indifferences. Note that it is not possible to apply the standard definition of anonymity because the set of admissible preferences differs from one agent to another. However, we could define a new property, ``type-anonymity", which would require that the outcome of the rule is independent of permutations of the agents with single-peaked preferences, and of permutations of the agents with single-dipped preferences. The introduction of this property would restrict the characterized family in Theorem \ref{theorem} to a subfamily that would include, among others, the rule introduced in Example \ref{ex1}. Moreover, note that we assume throughout the paper that preferences are linear orders. Allowing for indifferences would require that the preference domains include single-plateau \citep{moulin1984generalized, berga1998strategy} and single-basined \citep{bossert2014single} preferences, respectively.

\section*{Appendix} 

We introduce two preliminary results. 
First, since our domain ${\cal R}$ satisfies the condition of \emph{indirect sequential inclusion} introduced in \cite{barbera2010individual}, it follows from their Theorem 2 that on our domain, all SP rules are also GSP.\footnote{Example ($viii$) of Section 4.5 in \cite{barbera2010individual} mentions that our domain satisfies indirect sequential inclusion. The formal proof is available upon request.}

\begin{lemma}
\label{equivalence}
The social choice rule $f$ is SP if and only if it is GSP.
\end{lemma}


\noindent Next, we show that $f$ is independent of the preferences over alternatives that are not in the range of $f$.
To do that, given profile $R \in {\cal R}$ and a set $T \subset X $, let $R|_T$ be the restriction of $R$ to the set of alternatives $T$. 

\begin{lemma}
\label{independencia}
Let $f:\mathcal{R}\rightarrow\Omega$ be SP. 
Then, for each $R,R'\in{\cal R}$ such that $R|_{\Omega}=R'|_{\Omega}$, $f(R)=f(R')$.
\end{lemma}

\begin{proof}
Let $R,R'\in{\cal R}$ be such that $R|_{\Omega}=R'|_{\Omega}$. Suppose by contradiction that $f(R) \neq f(R')$. Starting at $R$, construct a sequence of profiles in which the preferences of all agents $i\in N$ are changed one by one from $R_i$ to $R'_i$ such that the sequence ends at $R'$. 
Since $f(R) \neq f(R')$, the outcome of the function must change along this sequence. 
Let $S \subset N$ be the set of agents that have changed preferences in the sequence the last time the rule selects $f(R)$. 
That is, $f(R'_{S}, R_{-S})=f(R)$. 
Let $i\in N\setminus S$ be the next agent changing preferences in the sequence. 
Then, by construction, $f(R'_{S\cup \{i\}}, R_{-(S\cup \{i\})}) \neq f(R'_{S}, R_{-S})$. 
If $f(R'_{S\cup \{i\}}, R_{-(S\cup \{i\})}) \, P_i \, f(R'_{S}, R_{-S})$, agent $i$ manipulates $f$ at $(R'_{S}, R_{-S})$ via $R'_i$. 
If, however, $f(R'_{S}, R_{-S}) \, P_i \, f(R'_{S\cup \{i\}}, R_{-(S\cup \{i\})})$, it follows from $R|_{\Omega}=R'|_{\Omega}$ that $f(R'_{S}, R_{-S}) \, P'_i \, f(R'_{S\cup \{i\}},$ $ R_{-(S\cup \{i\})})$. 
Agent $i$ then manipulates $f$ at $(R'_{S\cup \{i\}}, R_{-(S\cup \{i\})})$ via $R_i$.
\end{proof}

\noindent Given Lemma \ref{independencia}, we assume from now on that for each $R \in {\cal R}$, the peaks and dips coincide with the $\Omega$-restricted peaks and dips: $\rho(R) = p(R)$ and $\delta(R) = d(R)$.

\subsection*{Proof of Proposition \ref{barbera}}

Throughout the proof we will need the following notation. For each $i \in N$ and each $R_i \in {\cal R}_i$, we use the notation $o(R_i)$ to refer to $p(R_i)$ when $i \in A$ and to refer to $d(R_i)$ when $i \in D$. Similarly, $o(R)$ contains the vector of $\Omega$-restricted peaks and dips at profile $R$: $o(R) = (p(R), d(R))$. Then, $\Omega(o(R))$ contains the alternatives of $\Omega$ that appear as the outcome of $f$ for profiles such that the vector of $\Omega$-restricted peaks and dips is $o(R)$: $\Omega(o(R))=\{x \in \Omega : \exists R' \in {\cal R} \mbox{ such that } o(R') = o(R) \mbox{ and } f(R') = x\}$. 

\medskip

\noindent Let $R \in {\cal R}$. The proof is divided into five steps.

\bigskip

\noindent \textit{Step 1: We prove that $|\Omega(o(R))| \leq 2$.}

\medskip

\noindent \cite{alcalde2018strategy} analyzes the SP rules on a model in which the location of each agent is public information and the peak or dip of an agent's preference is at her location, but the social planner does not know if an agent has single-peaked or single-dipped preferences. Proposition 1 in \cite{alcalde2018strategy} shows that the range for the subdomain when all agents have declared their type of preference (single-peaked or single-dipped) contains at most two alternatives. Observe that in that subdomain, the types of preferences and the locations of all peaks and dips are known. In our model, types of preferences are publicly known. So, the subdomain that arises in our model once information about the locations of the peaks and dips, \emph{i.e.}, $(p(R), d(R))$, becomes available, is exactly the same as the one analyzed in Proposition 1 of \cite{alcalde2018strategy}. Hence, we must have that $|\Omega(o(R))| \leq 2$.

\bigskip

\noindent \textit{Step 2: We prove that if $|\Omega(o(R))|=2$, then for each $i \in D$ and each $R'_i \in {\cal R}_i$, $\Omega(o(R'_i, R_{-i})) = \Omega(o(R))$.}

\medskip

\noindent Suppose that $|\Omega(o(R))|=2$ and define $\min \Omega(o(R)) \equiv l < r \equiv \max\Omega(o(R))$. Given $R \in {\cal R}$, $N(R) = \{i \in N \, : \, o(R_i) \in (l, r)\}$ refers to the set of individuals who have their $\Omega$-restricted peaks/dips between $l$ and $r$ at profile $R$. The proof of Step 2 is based on intermediary results that we establish first. Lemma \ref{internal} shows that only the agents in $N(R)$ can affect the outcome of $f$.

\begin{lemma}
\label{internal}
Let $f: {\cal R} \rightarrow \Omega$ be SP. Then, for each $R, R' \in {\cal R}$ such that $o(R) = o(R')$, if $R_{N(R)} = R'_{N(R)}$, then $f(R) = f(R')$.
\end{lemma}

\begin{proof} Let $R,R'\in {\cal R}$ be such that $o(R)=o(R')$ and $R_{N(R)} = R'_{N(R)}$. Suppose by contradiction that $f(R) \neq f(R')$. Assume without loss of generality that $f(R) < f(R')$. Since $o(R)=o(R')$, $\Omega(o(R)) = \Omega(o(R'))$ by definition. And since $f(R) \neq f(R')$, $|\Omega(o(R))| \neq  1$. Consequently, it follows from Step 1 and $f(R) < f(R')$ that $f(R) = l$ and $f(R') = r$.

\medskip
 
\noindent Starting at $R$, construct a sequence of profiles in which the preferences of all agents $i \in N$ are changed one by one from $R_i$ to $R'_i$ so that the sequence ends at $R'$. 
In all profiles of the sequence, the vector of $\Omega$-restricted peaks and dips is the same and therefore, for each profile of the sequence, the outcome of $f$ is either $l$ or $r$. 
Since $f(R) \neq f(R')$, the outcome must have changed along the sequence. 
So, let $S \subset N$ be the set of agents that have changed preferences in the sequence the last time the rule selects $f(R)$, and let $i\in N$ be the next agent changing preferences in the sequence. 
Then, $f(R) = f(R'_S, R_{-S}) = l \neq  r = f(R'_{S \cup \{i\}}, R_{-(S \cup \{i\})}) = f(R')$. 
Since $R_{N(R)}=R'_{N(R)}$ by assumption, $i\notin N(R)$. 
If [$i \in A$ and $p(R_i) \leq  l$] or [$i \in D$ and $d(R_i) \geq r$], then $l \, P'_i \, r$.
Agent $i$ then manipulates $f$ at $(R'_{S\cup\{i\}}, R_{-(S\cup\{i\})})$ via $R_i$. 
Otherwise, if [$i \in A$ and $p(R_i) \geq r$] or [$i \in D$ and $d(R_i) \leq  l$], then $r \, P_i \, l$.
Agent $i$ then manipulates $f$ at $(R'_S, R_{-S})$ via $R'_i$. \end{proof}

\noindent Lemma \ref{internal} implies that if $|\Omega(o(R))| = 2$, then there is an agent whose $\Omega$-restricted peak/dip at profile $R$ is strictly between $l$ and $r$. This implies the following corollary.

\begin{corollary}
\label{medionovacio}
Let $f: {\cal R} \rightarrow \Omega$ be SP. Then, for each $R \in {\cal R}$ with $|\Omega(o(R))|=2$, $\Omega \cap (l,r) \neq \emptyset$.
\end{corollary}

\noindent Next, we show that if all agents of $N(R)$ prefer the same alternative from $\Omega(o(R))$, then that alternative has to be chosen.

\begin{lemma}
\label{unanimidad}
Let $f: {\cal R} \rightarrow \Omega$ be SP. Then, for each $R \in {\cal R}$ such that $|\Omega(o(R))| = 2:$
\begin{itemize}
\item[(i)] If $l \, P_i \, r$ for each $i \in N(R)$, then $f(R)= l$.
\item[(ii)] If $r \, P_i \, l$ for each $i \in N(R)$, then $f(R)= r$.
\end{itemize}
\end{lemma}

\begin{proof}
Due to symmetry reasons we only prove $(i)$.
Let $R \in {\cal R}$ be such that $\Omega(o(R))=\{l,r\}$ and for each $i \in N(R)$, $l\, P_i \, r$. 
Suppose by contradiction that $f(R)=r$. 
Since $l \in \Omega(o(R))$, there is a profile $R'\in {\cal R}$ such that $o(R')=o(R)$ and $f(R')=l$. 
By Lemma \ref{internal}, $f(R'_{N(R)}, R_{-N(R)}) = l$. Then, the set of agents $N(R)$ manipulates $f$ at $R$ via $R'_{N(R)}$.
Thus, $f$ is not GSP.
By Lemma \ref{equivalence}, $f$ is not SP.
\end{proof}

\noindent We now study deviations of agents with single-dipped preferences. 
In particular, given a profile $R\in \mathcal{R}$ such that $|\Omega(o(R))|=2$, we are interested in how $\Omega(o(R))$ changes as an agent $i \in D$ changes her preference from $R_i$ to $R'_i$. 
To do that, denote $\min {\Omega}(o(R'_i, R_{-i}))=l'$ and $\max {\Omega}(o(R'_i, R_{-i}))=r'$.

\begin{lemma}
\label{movementdips}
Let $f: {\cal R} \rightarrow \Omega$ be SP. Consider $i\in D$ and $R, (R'_i,R_{-i})\in{\cal R}$ such that $|\Omega(o(R))| = 2$. 
\begin{itemize}
\item[(i)] If $d(R_i) \leq l$, then $l' \in [d(R_i),l]$ and $r' \leq r$.

\item[(ii)] If $d(R_i) \geq r$, then $r' \in [r, d(R_i)]$ and $l' \geq l$.

\item[(iii)] If $i \in N(R)$, then $\{l',r'\} \subseteq [l,r]$. 
\end{itemize}
\end{lemma}

\begin{proof} Due to symmetry reasons we only prove $(i)$ and $(iii)$. Let $i\in D$ and $R, (R'_i, R_{-i}) \in {\cal R}$ be such that $\Omega(o(R))=\{l,r\}$ (with $l\neq r$) and $\Omega(o(R'_i,R_{-i}))=\{l',r'\}$ (possibly $l'=r'$). 

\medskip

\noindent Proof of $(i)$. 
Assume that $d(R_i)\leq l$. 
Suppose first by contradiction that $r' > r$. Consider a preference profile $\hat{R}' \in {\cal R}$ such that $o(\hat{R}') = o(R'_i, R_{-i})$ and $f(\hat{R}') = r'$.\footnote{It can be checked that all preference profiles and preference rankings introduced in the proofs exist. We omit these parts of the proofs, but they can be provided upon request.} 
Since $o(R_i, \hat{R}'_{-i}) = o(R)$, $f(R_i, \hat{R}'_{-i}) \in \{l, r\}$ and agent $i$ manipulates $f$ at this profile via $\hat{R}'_i$ to obtain $r'$. 
Thus, $r' \leq r$. 
Next, suppose by contradiction that $l' \not\in [d(R_i), l]$. If $l' > l$, consider $\hat{R} \in {\cal R}$ such that $o(\hat{R}) = o(R)$ and $f(\hat{R}) = l$. 
Note that $o(R'_i, \hat{R}_{-i}) = o(R'_i, R_{-i})$ and, therefore, $f(R'_i, \hat{R}_{-i}) \in \{l', r'\}$. Thus, agent $i$ manipulates $f$ at $\hat{R}$ via $R'_i$. 
Finally, if $l' < d(R_i)$, consider $\tilde{R}' \in {\cal R}$ such that $o(\tilde{R}') = o(R'_i, R_{-i})$ and $f(\tilde{R}') = l'$. 
Consider also a preference $\tilde{R}_i \in {\cal R}_i$ with $d(\tilde{R}_i) = d(R_i)$ and $l' \, \tilde{P}_i \, r$. Note that $o(\tilde{R}_i, \tilde{R}'_{-i}) = o(R)$ and, therefore, $f(\tilde{R}_i, \tilde{R}'_{-i}) \in \{l, r\}$. Thus, agent $i$ manipulates $f$ at this profile via $\tilde{R}'_i$ to obtain $l'$. 

\medskip

\noindent Proof of $(iii)$. 
Assume that $i\in N(R)$ and suppose by contradiction that $\{l',r'\} \not\subseteq [l,r]$.
Suppose without loss of generality that $l' \not\in [l, r]$. 
Consider $\bar{R}' \in {\cal R}$ such that $o(\bar{R}')=o(R'_i, R_{-i})$ and $f(\bar{R}') = l'$. Consider also $\bar{R}_i \in {\cal R}_i$ such that $d(\bar{R}_i) = d(R_i)$ and $l' \, \bar{P}_i \, w$ for each $w \in \{l, r\}$. Since $o(\bar{R}_i, \bar{R}'_{-i}) = o(R)$, we have that $f(\bar{R}_i, \bar{R}'_{-i}) \in \{l, r\}$. Thus, agent $i$ manipulates $f$ at this profile via $\bar{R}'_i$ to obtain $l'$.
\end{proof}

\noindent We are now ready to prove Step 2. The proof is divided into five cases. 
Consider $R\in{\cal R}$ such that $|\Omega(o(R))|=2$. 
We have to show that for each $i\in D$ and each $R'_i\in{\cal R}_{i}$, $\{l',r'\}=\{l,r\}$. 

\bigskip

\noindent \underline{Case 1}: $d(R_i) \in (l, r)$ and $d(R'_i) \in (l, r)$.

\medskip

\noindent Since $i \in N(R)$ by assumption, it follows from Lemma \ref{movementdips}\,$(iii)$ that $\{l', r'\} \subseteq [l, r]$. We distinguish three subcases. 
(a) If $d(R'_i) \leq l'$, we deduce from Lemma \ref{movementdips}\,$(i)$ that $l \in [d(R'_i), l']$ (note that $(R'_i, R_{-i})$ and $R_i$ play the roles of $R$ and $R'_i$, respectively). 
This contradicts that $d(R'_i) \in (l, r)$. 
(b) If $d(R'_i) \geq r'$, a similar contradiction is obtained by applying Lemma \ref{movementdips}\,$(ii)$. 
And (c), if $i \in N(R')$, we apply Lemma \ref{movementdips}\,$(iii)$ to see that $\{l, r\} \subseteq [l', r']$ (note that $(R'_i, R_{-i})$ and $R_i$ play the roles of $R$ and $R'_i$, respectively). This together with $\{l', r'\} \subseteq [l, r]$ implies that $l = l'$ and $r = r'$. This concludes the proof of Case 1.

\medskip

\noindent We introduce a claim that is needed later on.

\begin{claim}
\label{nose}
Let $f: {\cal R} \rightarrow \Omega$ be SP. Consider $i\in D$ and $R, (R'_i, R_{-i}) \in {\cal R}$ such that $\Omega(o(R)) = \{l, r\}$, $\Omega(o(R'_i, R_{-i})) = \{l', r\}$, $d(R_i) \leq l' < l < r$, and $d(R'_i) \in (l', r)$. If $N(R) \cap D \neq \emptyset$, then for each $j \in N(R) \cap D$ and each $R'_j \in {\cal R}_j$ such that $d(R'_j) \in (l', l]:$ 
$$\Omega(o(R'_{\{i,j\}}, R_{-\{i,j\}})) = \{l', r\}  \mbox{ and } \Omega(o(R'_j, R_{-j})) = \{\hat{l}, r\},\, \mbox{ where } \hat{l} \in [l, r).$$
\end{claim}

\begin{proof}
Let $i \in D$ and $R, (R'_i, R_{-i}) \in {\cal R}$ be such that $\Omega(o(R)) = \{l, r\}$, $\Omega(o(R'_i, R_{-i})) = \{l', r\}$, $d(R_i) \leq l' < l < r$, and $d(R'_i) \in (l', r)$. 
Suppose also that $N(R) \cap D \neq \emptyset$. Then, consider $j \in N(R) \cap D$ and $R'_j \in {\cal R}_j$ with $d(R'_j) \in (l', l]$. Denote $\Omega(o(R'_{\{i, j\}}, R_{-\{i, j\}}))=\{\bar{l},\bar{r}\}$ (possibly $\bar{l}=\bar{r}$) and $\Omega(o(R'_j, R_{-j}))=\{\hat{l},\hat{r}\}$ (possibly $\hat{l}=\hat{r}$). 

\medskip

\noindent We first show that $\{\bar{l}, \bar{r}\} = \{l', r\}$. 
It follows from $j \in N(R) \cap D$ and $l' < l$ that $d(R_j) \in (l', r)$.
Also, $d(R'_j) \in (l', r)$ by assumption. 
Thus, by Case 1 ---with $(R'_i,R_{-i})$ playing the role of $R$ and $j$ playing the role of $i$---, we have that $\bar{l}=l'$ and $\bar{r}=r$.

\medskip

\noindent Next, focus on $\{\hat{l}, \hat{r}\}$. 
By Lemma \ref{movementdips}\,$(iii)$ ---with $j$ playing the role of $i$---, $\{\hat{l}, \hat{r}\} \subseteq [l, r]$. 

\medskip

\noindent To show that $\hat{r}=r$, suppose by contradiction that $\hat{r}<r$. 
Consider $\tilde{R}\in {\cal R}$ such that $o(\tilde{R})=o(R)$ and $f(\tilde{R})=r$. 
Consider also $\tilde{R}'_j\in{\cal R}_j$ such that $d(\tilde{R}'_j)=d(R'_j)$ and for each $w\in\{\hat{l},\hat{r}\}$, $r \, \tilde{P}'_j \, w$. 
Since $o(\tilde{R}'_j, \tilde{R}_{-j}) = o(R'_j, R_{-j})$ and $\Omega(R'_j, R_{-j}) = \{\hat{l}, \hat{r}\}$, we have $\Omega(\tilde{R}'_j, \tilde{R}_j) = \{\hat{l}, \hat{r}\}$. Hence, $f(\tilde{R}'_j, \tilde{R}_j) \in \{\hat{l}, \hat{r}\}$ and agent $j$ manipulates $f$ at this profile via $\tilde{R}_j$. 

\medskip

\noindent To show that $\hat{l} \in [l, r)$, suppose by contradiction that $\hat{l} = r$. 

\noindent First, we show that for each $\bar{R} \in {\cal R}$ such that $o(\bar{R}) = o(R'_{\{i, j\}}, R_{-\{i, j\}})$ and $r \, \bar{P}_i \, l'$, $f(\bar{R}) = r$. Suppose by contradiction that there exists $\bar{R} \in {\cal R}$ such that $o(\bar{R}) = o(R'_{\{i, j\}}, R_{-\{i, j\}})$ and $r \, \bar{P}_i \, l'$, but $f(\bar{R}) = l'$. Since $o(R_i, \bar{R}_{-i}) = o(R'_j, R_{-j})$ and $\Omega(o(R'_j, R_{-j})) = r$,  $f(R_i, \bar{R}_{-i}) = r$.  Therefore, agent $i$ manipulates $f$ at $\bar{R}$ via $R_i$.

\noindent Second, we show that for each $\hat{R} \in {\cal R}$ such that $o(\hat{R}) = o(R'_i, R_{-i})$ and $r \, \hat{P}_i \, l'$, $f(\hat{R}) = r$. Suppose by contradiction that there exists $\hat{R} \in {\cal R}$ such that $o(\hat{R}) = o(R'_i, R_{-i})$ and $r \, \hat{P}_i \, l'$, but $f(\hat{R}) = l'$. Consider $\hat{R}'_j \in {\cal R}_j$ such that $d(\hat{R}'_j) = d(R'_j)$ and $l' \, \hat{P}'_j \, r$. Observe that $o(\hat{R}'_j, \hat{R}_{-j}) = o(R'_{\{i, j\}}, R_{-\{i, j\}})$. Then, $f(\hat R'_j,\hat R_{j})=r$ by the previous paragraph. Therefore, agent $j$ manipulates $f$ at this profile via $\hat R_j$.

\noindent Finally, consider $\tilde{R} \in{\cal R}$ such that $o(\tilde{R}) = o(R)$ and $f(\tilde{R})=l$. Also, let $\tilde{R}'_i\in{\cal R}_i$ be such that $d(\tilde{R}'_{i}) = d(R'_i)$ and $r \, \tilde{P}'_{i} \, l'$. Observe that $o(\tilde{R}'_i, \tilde{R}_{-i}) = o(R'_i, R_{-i})$. Then, $f(\tilde{R}'_i, \tilde{R}_{-i})=r$ by the previous paragraph. Since $d(\tilde{R}_i) = d(R_i) < l < r$, $r \, \tilde{P}_i \, l$ and agent $i$ manipulates $f$ at $\tilde{R}$ via $\tilde{R}'_{i}$. 

\noindent This concludes the proof of the claim.
\end{proof}

\noindent As a consequence of Claim \ref{nose}, for each $j\in D\setminus\{i\}$ and each $R'_j \in {\cal R}_j$ such that $d(R'_j) \in (l', l]$, $N(R'_j, R_{-j}) \cap D \subset N(R) \cap D$. 
The successive application of Claim \ref{nose} yields the following corollary.

\begin{corollary}
\label{claim}
Let $f: {\cal R} \rightarrow \Omega$ be SP. 
Consider $i\in D$ and $R, (R'_i, R_{-i}) \in {\cal R}$ such that $\Omega(o(R)) = \{l, r\}$, $\Omega(o(R'_i, R_{-i})) = \{l', r\}$, $d(R_i) \leq l' < l < r$, and $d(R'_i) \in (l', r)$. 
Then, there exists a profile $(R_i, \bar{R}_{-i}) \in {\cal R}$ such that $\Omega(o(R_i, \bar{R}_{-i})) = \{\hat{l}, r\}$, $\Omega(o(R'_i, \bar{R}_{-i})) = \{l', r\}$, with $\hat{l} \in [l, r)$, and $N(R_i, \bar{R}_{-i}) \cap D = \emptyset$.
\end{corollary}

\noindent We now prove the following lemma.

\begin{lemma}
\label{impossibility}
Let $f: {\cal R} \rightarrow \Omega$ be SP. Then, there is no $R, (R'_i, R_{-i}) \in {\cal R}$, with $i\in D$, such that $\Omega(o(R)) = \{l, r\}$, $\Omega(o(R'_i, R_{-i})) = \{l', r\}$, $d(R_i) \leq l' < l < r$, and $d(R'_i) \in (l', r)$. 
\end{lemma}

\begin{proof}
Suppose by contradiction that there exists $R, (R'_i, R_{-i}) \in {\cal R}$, with $i\in D$, such that $\Omega(o(R)) = \{l, r\}$, $\Omega(o(R'_i, R_{-i})) = \{l', r\}$, $d(R_i) \leq l' < l < r$, and $d(R'_i) \in (l', r)$. By Corollary \ref{claim}, there exists a profile $(R_i, \bar{R}_{-i}) \in {\cal R}$ such that $\Omega(o(R_i, \bar{R}_{-i})) = \{\hat{l}, r\}$, $\Omega(o(R'_i, \bar{R}_{-i})) = \{l', r\}$, with $\hat{l} \in [l, r)$, and $N(R_i, \bar{R}_{-i}) \cap D = \emptyset$. Since $|\Omega(o(R_i, \bar{R}_{-i}))| = 2$, we have by Lemma \ref{internal} that $N(R_i, \bar{R}_{-i}) \neq \emptyset$. 
This implies that $N(R_i, \bar{R}_{-i}) \subseteq A$. 
Consider a profile $\hat{R} \in {\cal R}$ such that ($i$) $o(\hat{R}) = o(R_i, \bar{R}_{-i})$, ($ii$) for each $j \in N(R_i, \bar{R}_{-i})$, $\hat{l} \, \hat{P}_j \, r \, \hat{P}_j \, l'$, and ($iii$) for each $k \in N$ such that $o(\hat{R}_k) \in (l', \hat{l}]$, $r \, \hat{P}_k \, l'$. 
Consider also $\hat{R}'_i \in {\cal R}_i$ such that $d(\hat{R}'_i) = d(R'_i)$ and $r \, \hat{P}'_i \, l'$. 
Observe that $o(\hat{R}'_i, \hat{R}_{-i}) = o(R'_i, \bar{R}_{-i})$ so that $\Omega(o(\hat{R}'_i, \hat{R}_{-i})) = \{l', r\}$. 
By Lemma \ref{unanimidad}, we have $f(\hat{R}) = \hat{l}$ and $f(\hat{R}'_i, \hat{R}_{-i}) = r$. 
Since $r \, \hat{P}_i \, \hat{l}$, agent $i$ manipulates $f$ at $\hat{R}$ via $\hat{R}'_i$.
\end{proof}

\noindent We are now ready to proceed with the remaining cases.

\medskip

\noindent \underline{Case 2}: [$d(R_i) \leq l$ and $d(R'_i) \leq l$] or [$d(R_i) \geq r$ and $d(R'_i) \geq r$].

\medskip

\noindent Due to symmetry reasons we only prove the case when $d(R_i) \leq l$ and $d(R'_i) \leq l$. 
Since by assumption $d(R_i) \leq l$, we have by Lemma \ref{movementdips}\,$(i)$ that $l' \in [d(R_i), l]$ and $r' \leq r$.

\medskip

\noindent Suppose first that $d(R'_i) < d(R_i)$. Since $d(R_i) \leq l'$, it follows that $d(R'_i) < l'$. 
Apply Lemma \ref{movementdips}\,$(i)$ ---with $(R'_i, R_{-i})$ and $R_i$ playing the roles of $R$ and $R'_i$, respectively--- to see that $l \in [d(R'_i), l']$ and $r \leq r'$. 
Therefore, $l = l'$ and $r = r'$. 

\medskip

\noindent Suppose next that $d(R_i) < d(R'_i)$. 
We first show that $d(R'_i) \leq l'$. 
Suppose by contradiction that $d(R_i) \leq l' < d(R'_i)$. 
We distinguish two cases. 
\begin{itemize}
\item If $r'>d(R'_i)$, then $i\in N(R'_i,R_{-i})$ and it follows from Lemma \ref{movementdips}\,$(iii)$ ---with $(R'_i,R_{-i})$ and $R_i$ playing the roles of $R$ and $R'_i$, respectively--- that $\{l,r\}\subseteq [l',r']$.
Since we already know that $r' \leq r$, $r'=r$.
However, this contradicts Lemma \ref{impossibility}.

\item If $r'\leq d(R'_i)$, apply Lemma \ref{movementdips}\,$(ii)$ ---with $(R'_i,R_{-i})$ and $R_i$ playing the roles of $R$ and $R'_i$, respectively--- to see that $r \in [r', d(R'_i)]$, which contradicts $d(R'_i) \leq l < r$.
\end{itemize}
We have shown that $d(R'_i) \leq l'$. Then, it follows from Lemma \ref{movementdips}\,$(i)$ ---with $(R'_i,R_{-i})$ and $R_i$ playing the roles of $R$ and $R'_i$, respectively--- that $l \in [d(R'_i), l']$ and $r \leq r'$. 
Therefore, $l = l'$ and $r = r'$. 

\medskip

\noindent \underline{Case 3}: $d(R_i) \not\in (l, r)$ and $d(R'_i) \in (l, r)$.

\medskip

\noindent Assume without loss of generality that $d(R_i) \leq l$. 
Lemma \ref{movementdips}\,$(i)$ implies that $l' \in [d(R_i), l]$ and $r' \leq r$. 
If $r' \leq d(R'_i)$, it follows from Lemma \ref{movementdips}\,$(ii)$ ---with $(R'_i, R_{-i})$ and $R_i$ playing the roles of $R$ and $R'_i$, respectively--- that $r \in [r', d(R'_i)]$. This contradicts the assumption $d(R'_i) \in (l, r)$.

\medskip

\noindent If $r' > d(R'_i)$, then $i \in N(R'_i,R_{-i})$.
Then, by Lemma \ref{movementdips}\,$(iii)$ ---with $(R'_i, R_{-i})$ and $R_i$ playing the roles of $R$ and $R'_i$, respectively---, $\{l, r\} \subseteq [l', r']$. 
Hence, $r' = r$. 
To finally show that $l' = l$, suppose by contradiction that $l' < l$. 
Hence, $d(R_i) \leq l' < l < d(R'_i) < r = r'$. However, this contradicts Lemma \ref{impossibility}.

\medskip

\noindent \underline{Case 4}: $d(R_i) \in (l,r)$ and $d(R'_i) \not\in (l, r)$.

\medskip

\noindent Assume without loss of generality that $d(R'_i) \leq l$. 
Since $i \in N(R)$, $\{l', r'\} \subseteq [l, r]$ by Lemma \ref{movementdips}\,$(iii)$.
It follows from $d(R'_i) \leq l$ and $l \leq l'$ that $d(R'_i) \leq l'$. 
Then, by Lemma \ref{movementdips}\,$(i)$ ---with $(R'_i, R_{-i})$ and $R_i$ playing the roles of $R$ and $R'_i$, respectively---, $r \leq r'$ and we conclude that $r' = r$. 
If $l' \in [d(R_i), r]$, then $d(R_i) \leq l'$ and $d(R'_i) \leq l'$. 
Apply Case 2 ---with $(R'_i, R_{-i})$ and $R_i$ playing the roles of $R$ and $R'_i$, respectively--- to see that $l = l'$.
If $l' \in [l, d(R_i))$, then $d(R_i) \in (l', r')$ and $d(R'_i) \not \in (l',r')$. 
Apply Case 3 ---with $(R'_i,R_{-i})$ and $R_i$ playing the roles of $R$ and $R'_i$, respectively--- to see that $l = l'$.
\medskip

\noindent\underline{Case 5}: [$d(R_i) \leq l$ and $d(R'_i) \geq r$] or [$d(R_i) \geq r$ and $d(R'_i) \leq l$].

\medskip

\noindent Due to symmetry reasons we only prove the case when $d(R_i) \leq l$ and $d(R'_i) \geq r$.
Since $|\Omega(o(R))|=2$, we have by Corollary \ref{medionovacio} that $\Omega \cap (l,r) \neq \emptyset$. 
Consider $\bar{R}_i \in {\cal R}_i$ such that $d(\bar{R}_i) \in (l, r)$. 
Apply Case 3 ---with $\bar{R}_i$ playing the role of $R'_i$--- to see that $\Omega(o(\bar{R}_i, R_{-i})) = \{l, r\}$. 
Finally, apply Case 4 ---with $(\bar{R}_i, R_{-i})$ playing the role of $R$--- to conclude that $l' = l$ and $r' = r$.

\bigskip

\noindent\textit{Step 3: We prove that if $|\Omega(o(R))|=1$, then for each $i \in D$ and $R'_i \in {\cal R}_i$, $|\Omega(o(R'_i,R_{-i}))|=1$.}

\medskip

\noindent Let $|\Omega(o(R))|=1$. 
Suppose by contradiction that there exists $i \in D$ and $R'_i \in {\cal R}_i$ such that $|\Omega(o(R'_i,R_{-i}))|\neq 1$. 
Then, by Step 1, $|\Omega(o(R'_i,R_{-i}))|= 2$. 
Thus, $\Omega(o(R'_i,R_{-i})) = \Omega(o(R))$ by Step 2 ---with $(R'_i,R_{-i})$ and $R_i$ playing the roles of $R$ and $R'_i$, respectively.
This contradicts that $|\Omega(o(R))|=1$.

\bigskip

\noindent\textit{Step 4: We prove that if $|\Omega(o(R))|=1$, then one of the following holds: (a) for each $R'_D \in {\cal R}^D$, $\Omega(o(R'_D, R_{-D}))=\Omega(o(R))$ or (b) there is an alternative $x\in \Omega \setminus \{\Omega(o(R))\}$ such that for each $R'_D \in {\cal R}^D$, either $\Omega(o(R'_D, R_{-D}))=\Omega(o(R))$ or $\Omega(o(R'_D, R_{-D}))=x$}.

\medskip

\noindent Let the mapping $h_f^R : {\cal R}^{D} \rightarrow \Omega$ be such that for each $R'_D \in {\cal R}^D$, $h_f^R(R'_D)=\Omega(o(R'_D, R_{A}))=f(R'_D, R_A)$. 
The mapping $h_f^R$ is well-defined because the successive application of Step 3 implies that for each $R'_D\in{\cal R}^D$, $|\Omega(o(R'_D, R_A))|=1$. 
Since $f$ is SP by assumption, $h_f^R$ must as well be SP. 
The domain of $h_f^R$ is the set of all profiles of single-dipped preferences. 
The result follows because \cite{barbera2012domains} shows for this domain that the range of $h_f^R$ contains at most two alternatives. 

\bigskip

\noindent \textit{Step 5: We prove that $|\Omega(p(R))| \leq 2$.}

\medskip

\noindent By Step 1, $|\Omega(o(R))|\leq 2$. 

\medskip

\noindent Suppose first that $|\Omega(o(R))|=2$ and define $\Omega(p(R)) = \Omega(o(R))$. 
We have to show that for each $R'\in{\cal R}$ such that $p(R')=p(R)$, $\Omega(o(R))=\Omega(o(R'))$. 
Consider $R'\in{\cal R}$ such that $p(R')=p(R)$. 
Starting at $R$, construct a sequence of profiles in which the preferences of the agents $i \in D$ are changed one by one from $R_i$ to $R'_i$ so that the sequence ends at $(R'_D,R_{A})$. 
Then, by successive applications of Step 2, we obtain that $\Omega(o(R'_D,R_A))=\Omega(o(R))$. 
Since $p(R')=p(R)$, we have that $o(R')=o(R'_D,R_A)$ and, therefore, $\Omega(o(R'))=\Omega(o(R'_D,R_A))$. 
Thus, $\Omega(o(R'))=\Omega(o(R))$.

\medskip

\noindent Suppose next that $|\Omega(o(R))|=1$. 
To prove the step for this case, we show that there exists $x \in \Omega \setminus \{\Omega(o(R))\}$ such that for each $R'\in{\cal R}$ with $p(R')=p(R)$, $[\Omega(o(R'))=\Omega(o(R))$ or $\Omega(o(R'))=x]$. 
Consider $R'\in{\cal R}$ such that $p(R')=p(R)$. 
By Step 4, there exists $x \in \Omega \setminus \{\Omega(o(R))\}$ such that either $\Omega(o(R'_D, R_A))=\Omega(o(R))$ or $\Omega(o(R'_D, R_A))=x$.
Since $p(R')=p(R)$, we have that $o(R')=o(R'_D,R_A)$ and, therefore, $\Omega(o(R'))=\Omega(o(R'_D,R_A))$. Thus, either $\Omega(o(R'))=\Omega(o(R))$ or $\Omega(o(R'))=x$. If $\Omega(o(R')) = x$ for some $R'\in{\cal R}$ with $p(R')=p(R)$, then we define $\Omega(p(R)) = \{\Omega(o(R)), x\}$. Otherwise, we define $\Omega(p(R)) = \Omega(o(R))$.

\subsection*{Proof of Proposition \ref{omega}}



\subsubsection*{A necessary lemma}

Corollary \ref{structure0} explains that each SP rule $f$ depends on a set of functions. The first of them, $\omega$, determines the set of alternatives that are preselected when the agents of $A$ have declared their $\Omega$-restricted peaks. That is, $\omega(\mathbf{p})$ gives the alternatives that can be selected by $f$ when the vector of $\Omega$-restricted peaks is equal to $\mathbf{p}$. This set of preselected alternatives includes at most two alternatives. The other functions are binary decision functions $\{g_{\omega(\mathbf{p})}\}_{\omega(\mathbf{p})\in \Omega^{2}}$ that allows to choose the winning alternative when there are two preselected.

\medskip

\noindent Based on these binary decision functions, we now define a binary decision function for each ${\bf p} \in \Omega^a$ such that $\omega({\bf p}) \in \Omega^2$ denoted by $g_{\bf{p}} : {\cal R} \rightarrow \omega({\bf p})$. Each of these functions is defined simply by $g_{\bf{p}}(R) = g_{\omega(\bf{p})}(R)$. We now introduce a particular structure for these functions. To do that, we first introduce the following notation.
We define $L(R)=\{i \in (N(R) \cap A) \cup D : \underline{\omega}(p(R)) \, P_i \, \overline{\omega}(p(R))\}$. Then, the class of voting by collections of left-decisive sets can be defined by specifying a set of coalitions $W(g_{\bf{p}}) \subseteq 2^N$, called left-decisive sets, such that $g_{\bf{p}}$ chooses $\underline{\omega}({\bf p})$ in a profile $R\in{\cal R}$ with $p(R)=\bf{p}$ if $L(R)$ belongs to $W(g_{\bf{p}})$, and $\overline{\omega}({\bf p})$ otherwise. We introduce the formal definition of these binary decision functions.

\begin{definition}
\label{manjunath}
Given $\mathbf{p} \in \Omega^a$ such that $\omega({\bf p}) \in \Omega^2$, the binary decision function $g_{\bf{p}}$ is called a voting by collections of left-decisive sets if there is a non-empty set of non-empty coalitions $W(g_{\bf{p}}) \subseteq 2^N$ such that for each $R \in {\cal R}$ with $p(R) = \bf{p}$,
	$$g_{\bf{p}}(R) = \left\{
	\begin{array}{ll}
	\underline{\omega}({\bf p}) & \mbox{ if } L(R) \in W(g_{\bf{p}}) \\*[5pt]
	\overline{\omega}({\bf p}) & \mbox{ otherwise, } 
	\end{array}
	\right. \vspace{0.5cm}$$
	and the following conditions are satisfied:
	\begin{itemize}
		\item $W(g_{p(R)}) \subseteq 2^{(N(R)\cap A)\cup D}$.
		\item If $B \in W(g_{p(R)})$ and $B \subset C \subseteq [(N(R)\cap A)\cup D]$, then $C \in W(g_{p(R)})$.
	\end{itemize}
\end{definition}

\noindent Definition \ref{manjunath} imposes some conditions on the left-decisive sets. Let $R\in\cal{R}$. The first condition requires that the left-decisive sets of a binary decision function $g_{p(R)}$ have to be subsets of $(N(R)\cap A)\cup D$. This condition implies that, once the $\Omega$-restricted peaks $p(R)$ are known, the decision between $\underline{\omega}(p(R))$ and $\overline{\omega}(p(R))$ only depends on the opinion of the agents with single-dipped preferences and those agents with single-peaked preferences whose $\Omega$-restricted peaks at $R$ are located between the two preselected alternatives. 
The second condition, a monotonicity property, says that all supersets of a left-decisive set are also left-decisive sets. 

\medskip

\noindent The following lemma shows that SP imposes that each binary decision function $g_{\bf{p}} : {\cal R} \rightarrow \omega({\bf p})$ should be a voting by collections of left-decisive sets.

\begin{lemma}
\label{structure2}
Let $f: \mathcal{R}\rightarrow\Omega$ be SP. Then, there is a function $\omega : \Omega^a \rightarrow \Omega \cup \Omega^2$ and a set of voting by collections of left-decisive sets $g_{{\bf p}} : {\cal R} \rightarrow \omega({\bf p})$ $($one for each ${\bf p} \in \Omega^a$ such that $\omega({\bf p}) \in \Omega^2)$ such that for each $R \in {\cal R}$ with $p(R)=\bf{p}$, $$f(R) = \left\{
	\begin{array}{ll}
	\omega({\bf p})  & \mbox{ if } \omega({\bf p}) \in \Omega \\*[5pt]
	
	g_{{\bf p}}(R) & \mbox{ if } \omega({\bf p}) \in \Omega^2.
	\end{array}
	\right.$$
\end{lemma}

\begin{proof} By Corollary \ref{structure0}, it only remains to be shown that $g_{\mathbf{p}}$ can be defined as a voting by collections of left-decisive sets. This is equivalent to show that there is a set of coalitions $W(g_{\bf{p}}) \subseteq 2^N$ that satisfies the conditions in Definition \ref{manjunath}. Given $\mathbf{p}\in\Omega^{a}$ such that $\omega({\bf p}) \in \Omega^2$, we define, for each $g_{\bf{p}}$, a set $W(g_{\bf{p}}) \subseteq 2^N$ in the following way: $B \in W(g_{\bf{p}})$ if there is $R \in {\cal R}$ such that $p(R)=\mathbf{p}$, $L(R) = B$, and $g_{\bf{p}}(R) = \underline{\omega}({\bf p})$. Observe that, by definition, we have that for each $R\in{\cal R}$, $W(g_{p(R)}) \subseteq 2^{(N(R)\cap A)\cup D}$.

\bigskip

\noindent \textit{Step 1: We show that if for some $\mathbf{p}\in \Omega^{a}$ such that $\omega({\bf p}) \in \Omega^2$, $B \in W(g_{\bf{p}})$, then for each $R' \in {\cal R}$ such that $p(R') = \bf{p}$ and $L(R') = B$, $g_{\bf{p}}(R') = \underline{\omega}({\bf p})$.}

\medskip

\noindent Let $\mathbf{p}\in \Omega^{a}$ such that $\omega({\bf p}) \in \Omega^2$ and $B \in W(g_{\bf{p}})$. Then, by definition, there is $R \in {\cal R}$ such that $p(R)=\bf{p}$, $L(R) = B$, and $g_{\bf{p}}(R) = \underline{\omega}({\bf p})$. Consider $\bar{R} \in {\cal R}$ such that $p(\bar{R}) = \bf{p}$ and $L(\bar{R}) = B$. Suppose by contradiction that $g_{\bf{p}}(\bar{R}) = \overline{\omega}({\bf p})$. Starting at $R$, construct a sequence of profiles in which the preferences of all agents $i \in [(N(R)\cap A)\cup D]$ are changed one by one from $R_i$ to $\bar{R}_i$ so that the sequence ends at $(\bar{R}_{(N(R)\cap A)\cup D}, R_{-((N(R)\cap A)\cup D)})$. Since $o(\bar{R}) = o(\bar{R}_{(N(R)\cap A)\cup D}, R_{-((N(R)\cap A)\cup D)})$, by Lemma \ref{internal} $f(\bar{R}) = f(\bar{R}_{(N(R)\cap A)\cup D}, R_{-((N(R)\cap A)\cup D)})$. We can therefore conclude that $g_{\bf{p}}(\bar{R}_{(N(R)\cap A)\cup D}, R_{-((N(R)\cap A)\cup D)}) = \overline{\omega}(\bf{p})$. It follows from $g_{\bf{p}}(R) \neq g_{\bf{p}}(\bar{R}_{(N(R)\cap A)\cup D}, R_{-((N(R)\cap A)\cup D)})$ that the outcome must have changed along the sequence. So, let $S \subset [(N(R)\cap A)\cup D]$ be the set of agents that have changed preferences in the sequence the last time $g_{\bf{p}}$ selects $\underline{\omega}({\bf p})$, and let $i\in [(N(R)\cap A)\cup D]$ be the next agent changing preferences in the sequence. Then, $g_{\bf{p}}(R) = g_{\bf{p}}(\bar{R}_S, R_{-S})= \underline{\omega}({\bf p}) \neq \overline{\omega}({\bf p}) = g_{\bf{p}}(\bar{R}_{S \cup \{i\}}, R_{-(S \cup \{i\})}) = g_{\bf{p}}(\bar{R}_{(N(R)\cap A)\cup D}, R_{-((N(R)\cap A)\cup D)})$. Therefore, $f(R) = f(\bar{R}_S, R_{-S}) = \underline{\omega}(p(R)) \neq  \overline{\omega}(p(R)) = f(\bar{R}_{S \cup \{i\}}, R_{-(S \cup \{i\})}) = f(\bar{R}_{(N(R)\cap A)\cup D},$ $R_{-((N(R)\cap A)\cup D)})$. If $i \in B$, $\underline{\omega}(p(R)) \, \bar{P}_i \, \overline{\omega}(p(R))$ and agent $i$ manipulates $f$ at $(\bar{R}_{S \cup \{i\}},$ $R_{-(S \cup \{i\})})$ via $R_i$. Otherwise, if $i \notin B$, $\overline{\omega}(p(R)) \, P_i \, \underline{\omega}(p(R))$ and agent $i$ manipulates $f$ at $(\bar{R}_S, R_{-S})$ via $\bar{R}_i$.

\bigskip

\noindent \textit{Step 2: We show that if for some $R\in{\cal R}$, $B \in W(g_{p(R)})$ and $B \subset C \subseteq [(N(R)\cap A)\cup D]$, then $C \in W(g_{p(R)})$.}

\medskip

\noindent Consider $R\in{\cal R}$, $B\in W(g_{p(R)})$, and $C \subseteq [(N(R)\cap A)\cup D]$ such that $B\subset C$. Suppose by contradiction that $C \notin W(g_{p(R)})$. Consider $R' \in {\cal R}$ such that $p(R') =p(R)$ and $L(R')=B$. Note that since $B \in W(g_{p(R)})=W(g_{p(R')})$, then, by Step 1, $g_{p(R)}(R')=\underline{\omega}(p(R))$ and, therefore, $f(R') = \underline{\omega}(p(R))$. Consider now $R''_{C\setminus B} \in {\cal R}^{C \setminus B}$ such that $p(R''_{C\setminus B}) = p(R_{C\setminus B})$ and for each $j\in C\setminus B$, $\underline{\omega}(p(R)) \, P''_j \, \overline{\omega}(p(R))$. Observe that $p(R''_{C\setminus B}, R'_{-(C \setminus B)}) = p(R)$ and, then $\omega(p(R''_{C\setminus B}, R'_{-(C \setminus B)})) = \omega(p(R))$. Since $L(R''_{C\setminus B}, R'_{-(C \setminus B)})=C\notin W(g_{p(R)})=W(g_{p(R''_{C\setminus B}, R'_{-(C \setminus B)})})$, we have $g_{p(R)}(R''_{C\setminus B}, R'_ {- (C\setminus B)})=\overline{\omega}(p(R))$. Thus, $f(R''_{C\setminus B}, R'_{-(C \setminus B)}) = \overline{\omega}(p(R))$. However, the agent set $C \setminus B$ manipulates $f$ at this profile via $R'_{C\setminus B}$ to obtain $\underline{\omega}(p(R))$ and $f$ is not GSP. By Lemma \ref{equivalence}, $f$ is not SP.
\bigskip

\noindent \textit{Step 3: We show that for each $\mathbf{p} \in \Omega^a$ such that $\omega({\bf p}) \in \Omega^2$, $\emptyset \notin W(g_{\bf{p}}) \neq \emptyset$.}

\medskip

\noindent Consider $\mathbf{p}\in \Omega^{a}$ such that $\omega({\bf p}) \in \Omega^2$. We only prove that $W(g_{\bf{p}}) \neq \emptyset$ because the other part is similar and thus omitted. To do it, we are going to show that, given $R \in {\cal R}$ such that $p(R) = \bf{p}$, and $L(R) = [(N(R)\cap A)\cup D]$, then $g_{\bf{p}}(R) = \underline{\omega}({\bf p})$ and, therefore, $[(N(R)\cap A)\cup D] \in W(g_{\bf{p}})$. Suppose otherwise that $g_{\bf{p}}(R) = \overline{\omega}({\bf p})$ and, therefore, $f(R) = \overline{\omega}({\bf p})$. Since $\underline{\omega}({\bf p}) \in \omega({\bf p})$, there is $R' \in {\cal R}$ such that $p(R') = \bf{p}$ and $f(R') = \underline{\omega}({\bf p})$. Then, starting at $R$, consider a sequence of profiles in which the preferences of all agents $i \in N$ are changed one by one from $R_i$ to $R'_i$ so that the sequence ends at $R'$. Observe that all profiles of the sequence have the same vector of $\Omega$-restricted peaks than $R$ and, therefore, $f$ chooses in each of them either $\underline{\omega}({\bf p})$ or $\overline{\omega}({\bf p})$. Since $f(R) \neq f(R')$, the outcome must have changed along the sequence. So, let $S \subset N$ be the set of agents that have changed preferences in the sequence the last time the rule selects $f(R)$, and let $i\in N$ be the next agent changing preferences in the sequence. Then, $f(R) = f(R'_S, R_{-S}) = \overline{\omega}({\bf p}) \neq  \underline{\omega}({\bf p}) = f(R'_{S \cup \{i\}}, R_{-(S \cup \{i\})}) = f(R')$. If $i\in L(R)$, then $\underline{\omega}({\bf p}) \, P_i \, \overline{\omega}({\bf p})$ and agent $i$ manipulates $f$ at $(R'_S, R_{-S})$ via $R'_i$. Otherwise, if $i\notin L(R)$, then $i\in A\setminus (N(R)\cap A)$. By Lemma \ref{internal}, $f(R'_S, R_{-S}) = f(R'_{S \cup \{i\}}, R_{-(S \cup \{i\})})$, which is a contradiction.
\end{proof}

\subsubsection*{First part of the proof of Proposition~\ref{omega} $(i)$}\label{firststepProp2}

\noindent This first part shows that if $\omega({\bf p})$ belongs to $\Omega^2$ for some vector of $\Omega$-restricted peaks ${\bf p}$, then these alternatives are contiguous in $\Omega$ and, thus, $\omega({\bf p}) \in \Omega^2_C$.

\medskip

\noindent A series of lemmas is needed to prove this statement.
Throughout these lemmas we consider $i \in A$ and $R, (R'_i,R_{-i}) \in {\cal R}$ such that $\omega(p(R))\in \Omega^2$. For simplicity, we denote $\underline{\omega}(p(R))\equiv l< r \equiv \overline{\omega}(p(R))$,  $l'\equiv\underline{\omega}(p(R'_i,R_{-i}))$, and $r' \equiv \overline{\omega}(p(R'_i,R_{-i}))$ (possible $l'=r'$). 

\medskip

\begin{lemma}\label{mediopeaks}
Let $f: {\cal R} \rightarrow \Omega$ be SP. If $p(R_i)\in (l,r)$, then $\{l', r'\} \cap (l, r) = \emptyset$.	
\end{lemma}

\begin{proof} Suppose by contradiction that $p(R_i) \in (l, r)$ but $l' \in (l, r)$.
The case when $r' \in (l, r)$ is similar and thus omitted. 
Let $\bar{R}' \in {\cal R}$ be such that $p(\bar{R}')=p(R'_i,R_{-i})$ and $f(\bar{R}')=l'$. 
Consider now $\bar{R}_i \in {\cal R}_i$ such that $p(\bar{R}_i)=p(R_i)$ and for each $v\in \{l,r\}$, $l' \, \bar{P}_i \, v$. 
Since $p(\bar{R}_i, \bar{R}'_{-i}) = p(R)$, $\omega(p(\bar{R}_i, \bar{R}'_{-i}))=\{l,r\}$. 
Thus, $f(\bar{R}_i, \bar{R}'_{-i}) \in \{l, r\}$ and agent $i$ manipulates $f$ at this profile via $\bar{R}'_i$ to obtain $l'$.
\end{proof}

\begin{lemma}
\label{doble}
Let $f: {\cal R} \rightarrow \Omega$ be SP. If $[p(R_i) \leq l$, $p(R'_i) \in (l, r)$, and $r' = r]$ or $[p(R_i) \geq r$, $p(R'_i) \in (l, r)$, and $l' = l]$, then $\{l', r'\} = \{l, r\}$.
\end{lemma}

\begin{proof} Due to symmetry reasons we only prove the case when $p(R_i)\leq l$, $p(R'_i) \in (l, r)$, and $r' = r$. 
Suppose by contradiction that $l' \neq l$.

\bigskip

\noindent \textit{Step 1: We show that $l' \in (l, p(R'_i)]$.}

\medskip

\noindent Suppose by contradiction that $l' \not\in (l, p(R'_i)]$. 
If $l' < l$, then $p(R'_i) \in (l',r')$.
Then, by Lemma \ref{mediopeaks} ---with $R'_i$ playing the role of $R_i$ and vice versa---, $l \not \in (l',r')$, which contradicts the assumption $l' < l$.
Otherwise, if  $l'> p(R'_i)$, consider $\bar{R}\in{\cal R}$ such that $p(\bar{R})=p(R)$ and $f(\bar{R})=l$. 
Also, consider $\bar{R}'_i\in{\cal R}_i$ such that $p(\bar{R}'_i)=p(R'_i)$ and $l \, \bar{P}'_i \, l'$.
Since $p(\bar{R}'_i, \bar{R}_{-i})=p(R'_i,R_{-i})$, $f(\bar{R}'_i, \bar{R}_{-i})\in \{l', r'\}$. 
Therefore, agent $i$ manipulates $f$ at this profile via $\bar{R}_i$ to obtain $l$. 

\bigskip

\noindent \textit{Step 2: We show that for each $C \in W(g_{p(R'_i,R_{-i})})$, $C \cap D \in W(g_{p(R)})$.}

\medskip

\noindent Let $C \in W(g_{p(R'_i,R_{-i})})$. 
Suppose by contradiction that $C \cap D \notin W(g_{p(R)})$. 
By Step 1, $l' \in (l, p(R'_i)]$. 
Consider $\hat{R} \in {\cal R}$ such that

\begin{enumerate}
\item[($i$)] $p(\hat{R}) = p(R'_i,R_{-i})$,

\item[($ii$)] for each $j \in [(N(\hat{R})\cap A)\cup D]$, $l' \, \hat{P}_j \, r$ if and only if $j \in C$, and 

\item[($iii$)] for each $k \in [N(R)\cap A] \cup [D\setminus C]$, $r \, \hat{P}_k \, l$.
\end{enumerate}

\noindent Then, $\omega(p(\hat{R}))=\{l',r\}$ and $L(\hat{R}) = C$. 
Observe that for each $j \in D$ such that $l' \, \hat{P}_j \, r$, we also have that $l \, \hat{P}_j \, r$. 
Since $p(R_i, \hat{R}_{-i}) = p(R)$, $\omega(p(R_i, \hat{R}_{-i}))=\{l,r\}$. 
By construction, $L(R_i, \hat{R}_{-i}) = C \cap D$. 
Since $C \in W(g_{p(R'_i,R_{-i})}) = W(g_{p(\hat{R})})$ and $C \cap D \notin W(g_{p(R)}) = W(g_{p(R_i, \hat{R}_{-i})})$, we have that $f(\hat{R}) = l'$ and $f(R_i, \hat{R}_{-i}) = r$. 
Finally, $p(R_i) \leq l < l' < r$ implies that $l' \, P_i \, r$. 
So, agent $i$ manipulates $f$ at $(R_i, \hat{R}_{-i})$ via $\hat{R}_i$.

\bigskip

\noindent \textit{Step 3: We show that there exists $B \in W(g_{p(R'_i,R_{-i})})$ such that $B \cap D = \emptyset$.}

\medskip

\noindent By Step 1, $l' \in (l, p(R'_i)]$. The proof is divided into two cases depending on $l'$. 

\begin{itemize}
\item Suppose that $l' = p(R'_i)$. 
By Lemma \ref{structure2} (exactly by the non-emptyness requirement in Definition \ref{manjunath}), $W(g_{p(R)}) \neq \emptyset$. 
Let $C \in W(g_{p(R)})$. 
Observe that $[C \cap N(R'_i, R_{-i}) \cap A] \cap D = \emptyset$. 
We complete the proof by showing that $[C \cap N(R'_i, R_{-i}) \cap A] \in W(g_{p(R'_i, R_{-i})})$.
Suppose by contradiction that $[C \cap N(R'_i, R_{-i}) \cap A] \notin W(g_{p(R'_i, R_{-i})})$. 
Consider $\tilde{R} \in {\cal R}$ such that
\begin{enumerate}
\item[($i$)] $p(\tilde{R}) = p(R)$,

\item[($ii$)] for each $j \in (N(\tilde{R})\cap A)$, $l \, \tilde{P}_j \, r$ if and only if $j \in C$,

\item[($iii$)] for each $k \in D$, $d(\tilde{R}_k)=l'$ and [$l \, \tilde{P}_k \, r$ if and only if $k \in C$], and

\item[($iv$)] for each $m \in [N(R'_i, R_{-i}) \cap  (A \setminus (C\cap A))]$, $r \, \tilde{P}_m \, l'$.
\end{enumerate}

Then, $\omega(p(\tilde{R}))=\{l,r\}$ and $L(\tilde{R}) = C$. 
Consider  $\tilde{R}'_i \in {\cal R}_i$ such that $p(\tilde{R}'_i) = p(R'_i)$ and $l \, \tilde{P}'_i \, r$. 
Observe that for each $j\in N(R'_i,R_{-i})\cap A$ such that $l \, \tilde{P}_j \, r$, we have that $l' \, \tilde{P}_j \, r$. 
Since $p(\tilde{R}'_i, \tilde{R}_{-i}) = p(R'_i, R_{-i})$, $\omega(p(\tilde{R}'_i, \tilde{R}_{-i}))=\{l',r\}$. 
By construction, $L(\tilde{R}'_i, \tilde{R}_{-i}) = C \cap N(R'_i, R_{-i}) \cap A$. 
Since $C \in W(g_{p(R)}) = W(g_{p(\tilde{R})})$ and $[C \cap N(R'_i, R_{-i}) \cap A] \notin W(g_{p(R'_i,R_{-i})}) = W(g_{p(\tilde{R}'_i,\tilde{R}_{-i})})$, we have that $f(\tilde{R}) = l$ and $f(\tilde{R}'_i, \tilde{R}_{-i}) = r$. 
Therefore, agent $i$ manipulates $f$ at $(\tilde{R}'_i, \tilde{R}_{-i})$ via $\tilde{R}_i$.

\item Suppose that $l' \in (l, p(R'_i))$. 
We complete the proof by showing that $\{i\} \in W(g_{p(R'_i, R_{-i})})$.
Suppose by contradiction that $\{i\} \notin W(g_{p(R'_i, R_{-i})})$.
Consider $R'' \in {\cal R}$ such that
\begin{enumerate}
\item[($i$)] $p(R'') = p(R'_i,R_{-i})$,

\item[($ii$)] $l' \, P''_i \, l \, P''_i \, r$,

\item[($iii$)] for each $j \in D$, $l \, P''_j \, r \, P''_j \, l'$,

\item[($iv$)] for each $k \in [(N(R'_i,R_{-i})\cap A) \setminus \{i\}] $, $r \, P''_k \, l'$, and

\item[($v$)] for each $m \in [(N(R) \setminus N(R'_i,R_{-i}))\cap A]$, $l \, P''_m \, r$.
\end{enumerate}

Then, $\omega(p(R''))=\{l',r\}$ and $L(R'')=\{i\}$. 
Since $p(R_i, R''_{-i})=p(R)$, we have that $\omega(p(R_i, R''_{-i}))=\{l,r\}$. 
By construction, $L(R_i, R''_{-i}) = [(N(R) \setminus N(R'_i,R_{-i}))\cap A]\cup D$. 
Since $\{i\} \notin W(g_{p(R'_i,R_{-i})})=W(g_{p(R'')})$, we have that $f(R'')=r$. 
We know from Lemma \ref{structure2} ---exactly from the non-emptiness requirement and the second point in Definition \ref{manjunath}--- that $[(N(R'_i,R_{-i}) \cap A) \cup D]\in W(g_{p(R'')}) = W(g_{p(R'_i, R_{-i})})$. 
By Step 2, $[(N(R'_i,R_{-i}) \cap A )\cup D] \cap D =D \in W(g_{p(R)}) = W(g_{p(R_i, R''_{-i})})$. 
Given that $ D \subseteq [(N(R) \setminus N(R'_i,R_{-i}))\cap A]\cup D$, it follows from Lemma \ref{structure2} ---exactly from the second point in Definition \ref{manjunath}--- that $[(N(R) \setminus N(R'_i,R_{-i}))\cap A]\cup D\in W(g_{p(R_i, R''_{-i})})$. Then, $f(R_i, R''_{-i})=l$. 
Therefore, agent $i$ manipulates $f$ at $R''$ via $R_i$.
\end{itemize}

\noindent \textit{Step 4: We find a contradiction.}

\medskip

\noindent By Step 3, there exists $B \in W(g_{p(R'_i,R_{-i})})$ such that $B \cap D = \emptyset$. 
Then, by Step 2 ---with $B$ playing the role of $C$--- we have that $B \cap D \in W(g_{p(R)})$. 
Therefore, $\emptyset \in W(g_{p(R)})$. 
Since $\omega(p(R)) \in \Omega^2$, this contradicts Lemma \ref{structure2} (exactly the non-emptiness requirement in Definition \ref{manjunath}). \end{proof}

\begin{lemma}
\label{nosalto}
Let $f: {\cal R} \rightarrow \Omega$ be SP.
\begin{itemize}
\item[$(i$)] If $p(R_i) \leq l$, then $l' \geq l$ and $r' \geq r$.
\item[$(ii$)] If $p(R_i) \geq r$, then $l' \leq l$ and $r' \leq r$.
\end{itemize}
\end{lemma}

\begin{proof} Due to symmetry reasons we only prove ($i$).
Suppose that $p(R_i)\leq l$. 
We divide the proof into two steps.

\bigskip

\noindent \textit{Step 1: We show that if $l' \geq l$, then $r' \geq r$.}

\medskip

\noindent Suppose by contradiction that $p(R_i) \leq l\leq l'$ but $r' < r$. 
Consider $\bar{R} \in{\cal R}$ such that $p(\bar{R})=p(R)$ and $f(\bar{R}) = r$. 
Since $p(R'_i, \bar{R}_{-i}) = p(R'_i, R_{-i})$, $f(R'_i, \bar{R}_{-i}) \in \{l', r'\}$. 
It follows from $p(\bar{R}_i) \leq l' \leq r' < r$ that for each $w \in \{l', r'\}$, $w \, \bar{P}_i \, r$. 
Then, agent $i$ manipulates $f$ at $\bar{R}$ via $R'_i$.

\bigskip

\noindent \textit{Step 2: We show that $l' \geq l$.}

\medskip

\noindent Suppose by contradiction that $l' < l$. 

\medskip

\noindent First, suppose that $p(R_i) < l$ and consider $\hat{R}' \in {\cal R}$ such that $p(\hat{R}')=p(R'_i,R_{-i})$ and $f(\hat{R}')=l'$. 
Also, let $\hat{R}_i \in {\cal R}_i$ be such that $p(\hat{R}_i) = p(R_i)$ and $l' \, \hat{P}_i \, l$. 
Since $p(\hat{R}_i, \hat{R}'_{-i})=p(R)$, $f(\hat{R}_i, \hat{R}'_{-i}) \in \{l,r\}$. 
Agent $i$ manipulates $f$ at this profile via $\hat{R}'_i$ to obtain $l'$. 
This contradicts that $p(R_i) < l$. Then, assume from now on that $p(R_i) = l$.

\medskip
 
\noindent If $r' < r$, then consider $\tilde{R} \in {\cal R}$ such that
\begin{itemize}
\item[($i$)] $p(\tilde{R}) = p(R)$,

\item[($ii$)] $l' \, \tilde{P}_i \, r$, and

\item[($iii$)] for each $j \in [(N(R)\cap A)\cup D]$, $r \, \tilde{P}_j \, l$.
\end{itemize}

\noindent Then, $\omega(p(\tilde{R}))=\{l,r\}$ and $L(\tilde{R}) = \emptyset$. 
By Lemma \ref{structure2} ---exactly by the non-emptiness requirement in Definition \ref{manjunath}---, $f(\tilde{R}) = r$. 
Since $p(R'_i, \tilde{R}_{-i})=p(R'_i, R_{-i})$, $f(R'_i, \tilde{R}_{-i}) \in \{l',r'\}$. 
Finally, since $w \, \tilde{P}_i \, r$ for each $w \in \{l', r'\}$, agent $i$ manipulates $f$ at $\tilde{R}$ via $R'_i$. This contradicts that $r' < r$. Then, assume from now on that $r \leq r'$. We consider three subcases.


\begin{itemize}
\item Suppose that $p(R'_i) \in (l', r')$.
It follows from Lemma \ref{mediopeaks} ---with $R'_i$ playing the role of $R_i$ and vice versa--- that $\{l, r\} \cap (l', r') = \emptyset$.
This contradicts that $l \in (l', r')$.

\item Suppose that $p(R'_i) \geq r'$.
Consider $R''\in{\cal R}$ such that $p(R'')=p(R'_i,R_{-i})$ and $f(R'')=l'$. 
Since $p(R_i,R''_{-i})=p(R)$, $f(R_i,R''_{-i}) \in \{l,r\}$. 
Since $v\, P''_i \, l'$ for each $v\in\{l,r\}$, agent $i$ manipulates $f$ at $R''$ via $R_i$.

\item Suppose that $p(R'_i) \leq l'$.
Then, we have that $p(R'_i)\leq l'<l$. 
It follows from Step 1 ---with $R'_i$ playing the role of $R_i$ and vice versa--- that $r \geq r'$. 
Thus, $r' = r$. 
Observe that $p(R'_i) \leq l'$, $p(R_i) \in (l', r')$, and $r = r'$. 
By Lemma \ref{doble} ---with $R'_i$ playing the role of $R_i$ and vice versa---, $\{l, r\} = \{l', r'\}$.
This contradicts that $l' < l$.
\end{itemize}
This concludes the proof of the lemma.  \end{proof}

\begin{lemma}
\label{movepeaks}
Let $f: {\cal R} \rightarrow \Omega$ be SP. If $p(R'_i) \in (l, r)$, then $\{l', r'\} = \{l, r\}$.
\end{lemma}

\begin{proof}
First, if $r' \leq p(R'_i)$, apply Lemma \ref{nosalto}\,($ii$) ---with $R'_i$ playing the role of $R_i$ and vice versa--- to obtain that $l \leq l'$ and $r \leq r'$. This contradicts that $r' \leq p(R'_i) < r$.  
If $l' \geq p(R'_i)$, a similar contradiction is reached by applying Lemma \ref{nosalto}\,$(i)$. 
Hence, $p(R'_i)\in(l',r')$. 
Apply Lemma \ref{mediopeaks} ---with $R'_i$ playing the role of $R_i$ and vice versa--- to see that $\{l,r\} \cap (l', r')=\emptyset$; {\it i.e.}, $\{l',r'\}\subseteq [l,r]$.

\medskip

\noindent If $p(R_i) \in (l, r)$, we have by Lemma \ref{mediopeaks} that $\{l',r'\}\cap (l,r)=\emptyset$. 
Thus, $l' \leq l$ and $r' \geq r$. 
Since $\{l',r'\}\subseteq [l,r]$, we conclude that $\{l', r'\} = \{l, r\}$. 
\medskip

\noindent If $p(R_i) \not\in (l, r)$, we can assume without loss of generality that $p(R_i) \leq l$. 
By Lemma \ref{nosalto}\,$(i)$, $r'\geq r$. 
Since $\{l',r'\}\subseteq [l,r]$, $l' \in [l, p(R'_i))$ and $r' = r$.
Then, $p(R_i)\leq l$, $p(R'_i)\in(l,r)$, and $r'=r$. 
Apply Lemma \ref{doble} to see that $\{l',r'\}=\{l,r\}$.
\end{proof}

\noindent We are now ready to prove this part. 
So, let ${\bf p} \in \Omega^a$ such that $\omega({\bf p}) \in \Omega^2$ and suppose by contradiction that there is an alternative $x \in \Omega \cap (\underline{\omega}({\bf p}), \overline{\omega}({\bf p}))$.  
Then, there exists a profile $\bar{R} \in {\cal R}$ such that $x \in \omega(p(\bar{R}))$. Consider also a profile $R \in {\cal R}$ such that $p(R)=\textbf{p}$ and a subprofile $\hat{R}_A \in {\cal R}^A$ such that for each $i \in A$, $p(\hat{R}_i) = x$. 
Starting at $R$, construct a sequence of profiles in which the preferences of all agents $i\in A$ are changed one by one from $R_i$ to $\hat{R}_i$ so that the sequence ends at $(\hat{R}_A,R_{D})$. 
It follows from the successive application of Lemma \ref{movepeaks} that $\omega(p(\hat{R}_A, R_D)) = \{\underline{\omega}({\bf p}), \overline{\omega}({\bf p})\}$. 
Consider now a profile $\tilde{R}\in{\cal R}$ such that $p(\tilde{R})=p(\bar{R})$ and $f(\tilde{R})=x$. 
Since $p(\hat{R}_A, \tilde{R}_D)=p(\hat{R}_A,R_D)$, $f(\hat{R}_A, \tilde{R}_D) \in \{\underline{\omega}({\bf p}), \overline{\omega}({\bf p})\}$. 
Therefore, the agent set $A$ manipulates $f$ at this profile via $\tilde{R}_A$ to obtain $x$. 
Hence, $f$ is not GSP.
By Lemma \ref{equivalence}, $f$ is not SP.

\subsubsection*{Second part of the proof of Proposition~\ref{omega} $(i)$}\label{seconstepProp2}

\noindent We now show that for each interior alternative $x$ of $\Omega$, there is a vector of $\Omega$-restricted peaks ${\bf p}$ such that only $x$ is preselected by $\omega$, \emph{i.e.}, $\omega({\bf p}) = x$.

\medskip

\noindent Consider $x \in \textbf{int}(\Omega)$ and suppose by contradiction that for each $R \in {\cal R}$, $\omega(p(R)) \neq x$. 
Since $x \in \Omega$, there exists $R' \in {\cal R}$ such that for some $y \in \Omega \setminus \{x\}$, $\omega(p(R')) = \{x,y\}$. 
Assume without loss of generality that $x < y$. By the first part of the proof, $(x, y) \cap \Omega = \emptyset$. Since $x \in \textbf{int}(\Omega)$, there exists $z \in \Omega$ such that $z < x$. 
Then, there exists $\bar{R} \in {\cal R}$ such that $z \in \omega(p(\bar{R}))$.

\medskip

\noindent Consider $\tilde{R} \in {\cal R}$ such that
\begin{itemize}
\item[($i$)] for each $i \in A$, $p(\tilde{R}_i)=x$ and for each $w \leq x$, $w \, \tilde{P}_i \, y$, and

\item[($ii$)]for each $j \in D$, $d(\tilde{R}_j)=y$.
\end{itemize}

\noindent We first show that $f(R'_A, \tilde{R}_D) = x$. 
Since $p(R'_A, \tilde{R}_D) = p(R')$, $\omega(p(R'_A, \tilde{R}_D)) = \{x,y\}$. 
Observe that by construction, $N(R'_A, \tilde{R}_D)\cap A=\emptyset$ and $L(R'_A, \tilde{R}_D) = D$. Then, by Lemma \ref{structure2} ---exactly by a combination of the non-emptiness requirement and the second point in Definition \ref{manjunath}---, we have that $f(R'_A, \tilde{R}_D) = x$.

\medskip

\noindent Next, we show that $x \in \omega(p(\tilde{R}))$. 
Suppose by contradiction that $x \notin \omega(p(\tilde{R}))$. 
Then, $f(\tilde{R}) \neq x$ and the agent set $A$ manipulates $f$ at $\tilde{R}$ via $R'_A$ to obtain $x$.
Hence $f$ is not GSP. 
By Lemma \ref{equivalence}, $f$ is not SP.

\medskip

\noindent Since $\omega(p(\tilde{R})) \neq x$ by assumption, it follows from  the first part of this proof that $\omega(p(\tilde{R})) \in \{\{w, x\}, \{x, y\}\}$, with $w < x$ and $(w, x) \cap \Omega = \emptyset$. Suppose first that $\omega(p(\tilde{R})) = \{ w, x\}$. 
Observe that by construction, $N(\tilde{R})\cap A=\emptyset$ and $L(\tilde{R})=D$. 
Then, by Lemma \ref{structure2} ---exactly by a combination of the non-emptiness requirement and the second point in Definition \ref{manjunath}---, $f(\tilde{R}) = w$. 
The agent set $A$ thus manipulates $f$ at this profile via $R'_A$ to obtain $x$.
Hence, $f$ is not GSP. 
By Lemma \ref{equivalence}, $f$ is not SP.

\medskip

\noindent Finally, suppose that $\omega(p(\tilde{R})) = \{x, y\}$. Consider $\tilde{R}'\in {\cal R}$ such that
\begin{itemize}
\item[($i$)] $\tilde{R}'_A = \tilde{R}_A$, and

\item[($ii$)] for each $j \in D$, $d(\tilde{R}'_j)=x$.
\end{itemize}

\noindent Since $p(\tilde{R}') = p(\tilde{R})$, we have that $\omega(p(\tilde{R}')) = \{x, y\}$. 
Observe that by construction,  $N(\tilde{R}')\cap A=\emptyset$ and $L(\tilde{R}') = \emptyset$. 
Then, by Lemma \ref{structure2} ---exactly by the non-emptiness requirement in Definition \ref{manjunath}---, $f(\tilde{R}')=y$. 
Consider now $(\bar{R}_A,\tilde{R}'_D)\in{\cal R}$. 
Since $p(\bar{R}_A,\tilde{R}'_D)=p(\bar{R})$, it must be the case that $\omega(p(\bar{R}_A,\tilde{R}'_D))=\omega(p(\bar{R}))$. Given that $z\in \omega(p(\bar{R}))$, we have by the first part of this proof that $\omega(p(\bar{R}_A,\tilde{R}'_D))\in\{z, \{u,z\}\}$, where  $u \leq x$ and $u \neq z$. 
Therefore, $f(\bar{R}_A,\tilde{R}'_D) \leq x$. 
Since for each $i\in A$ and each $w\leq x$, $w\, \tilde{P}'_i\, y$ , the agent set $A$ manipulates $f$ at $\tilde{R}'$ via $\bar{R}_A$.
Hence, $f$ is not GSP. By Lemma \ref{equivalence}, $f$ is not SP.

\medskip

\subsubsection*{Proof of Proposition~\ref{omega} $(ii)$ and $(iii)$}

We only prove ($ii$) because the proof of $(iii)$ is similar. Suppose that $\min \Omega$ exists, but there is no $x \in \Omega$ such that $(\min \Omega, x) \in r_{\omega}$. If $\min \Omega \notin r_{\omega}$, then $\min \Omega \not\in \Omega$, which is a contradiction. 

\subsection*{Proof of Proposition \ref{generalized}}

We first define $r_\omega$. We include $\min \Omega$ in $r_\omega$ if and only if $f(R)=\min \Omega$ for all $R \in {\cal R}$ such that $p(R_i) = \min \Omega$ for each $i \in A$.
Similarly, we include $\max \Omega$ in $r_\omega$ if and only if $f(R)=\max \Omega$ for all $R \in {\cal R}$ such that $p(R_i) = \max \Omega$ for each $i \in A$.
Next, given $r_\omega$, define a correspondence $\mathcal{L} : r_\omega \rightarrow 2^A$ such that for each $\alpha \in r_\omega$, $C \in {\cal L}(\alpha)$ if [$C\in {\cal L}(\beta)$ for some $\beta <^* \alpha$] or [there is $R \in {\cal R}$ such that ($p(R_i) \leq^* \alpha \Leftrightarrow i\in C$) and $\omega(p(R))= \alpha$].

\bigskip

\noindent We establish the following lemma.

\begin{lemma}\label{lemmamedian}
Let $f: {\cal R} \rightarrow \Omega$ be SP. Consider $\alpha \in r_\omega$ and $C \subseteq 2^A$ such that for each $\beta <^* \alpha$, $C \in {\cal L}(\alpha)\setminus {\cal L}(\beta)$. 
\begin{itemize}
\item[$(i)$] If $R'\in{\cal R}$ is such that $p(R'_i) \leq^* \alpha$ for each $i \in C$, then $\omega(p(R')) \leq^* \alpha$.
\item[$(ii)$] If $R' \in {\cal R}$ is such that $p(R'_i) >^* \alpha$ for each $i \in A \setminus C$, then $\omega(p(R'))\geq^{*}\alpha$.
\end{itemize}
\end{lemma}

\noindent \begin{proof} Due to symmetry reasons we only prove $(i)$. 
Let $\alpha \in r_\omega$ and consider $C \subseteq 2^A$ such that for each $\beta <^* \alpha$, $C \in {\cal L}(\alpha) \setminus {\cal L}(\beta)$. Then, by definition of $\cal{L}(\alpha)$, there is $R \in {\cal R}$ such that $[p(R_i) \leq^* \alpha \Leftrightarrow i\in C]$ and $\omega(p(R))= \alpha$. 
Consider $R'\in{\cal R}$ such that $p(R'_i)\leq^{*}\alpha$ for each $i \in C$, and suppose by contradiction that $\omega(p(R')) >^* \alpha$. 
Let $C' = \{i\in A: p(R'_i) \leq^* \alpha\}$. Note that $C \subseteq C'$. Consider now $\bar{R}, \bar{R}', \hat{R}' \in {\cal R}$ such that
\begin{enumerate}
\item[($i$)] $p(\bar{R})=p(R)$ and $p(\bar{R}') = p(\hat{R}') = p(R')$,
\item[($ii$)] $f(\bar{R})=\underline{\omega}(p(R))$ and $f(\hat{R}') = \overline{\omega}(p(R'))$,
\item[($iii$)] for each $j \in A$ and each $v, w \in \Omega$ with $v \leq^* \alpha <^* w$, [$v \, \bar{P}_j \, w \Leftrightarrow j \in C$], and
\item[($iv$)] for each $j \in A$ and each $v, w \in \Omega$ with $v \leq^* \alpha <^* w$, [$v \, \bar{P}'_j \, w \Leftrightarrow v \, \hat{P}'_j \, w \Leftrightarrow j \in C'$].
\end{enumerate}

\noindent We first show that $f(\bar{R}'_A, \bar{R}_D) \leq \underline{\omega}(p(R))$. 
Suppose by contradiction that $f(\bar{R}'_A, \bar{R}_D) > \underline{\omega}(p(R))$. 
Starting at $\bar{R}$, construct a sequence of profiles in which the preferences of all agents $i\in A$ are changed one by one from $\bar{R}_i$ to $\bar{R}'_i$ such that the sequence ends at $(\bar{R}'_A, \bar{R}_D)$.
Since $f(\bar{R})=\underline{\omega}(p(R))< f(\bar{R}'_A, \bar{R}_D)$, there exists $S \subset A$ and $i \in (A \setminus S)$ such that $f(\bar{R}'_{S}, \bar{R}_{-S})\leq \underline{\omega}(p(R))$ and $f(\bar{R}'_{S\cup \{i\}}, \bar{R}_{-(S\cup \{i\})})> \underline{\omega}(p(R))$. 
If $i\in C'$, then, by construction, $v \, \bar{P}'_i \, w$ for each $v, w \in \Omega$ with $v \leq^* \alpha <^* w$. 
Thus, $f(\bar{R}'_S, \bar{R}_{-S}) \, \bar{P}'_i \, f(\bar{R}'_{S\cup \{i\}}, \bar{R}_{-(S\cup \{i\})})$.
Agent $i$ then manipulates $f$ at $(\bar{R}'_{S\cup \{i\}}, \bar{R}_{-(S\cup \{i\})})$ via $\bar{R}_i$. 
Otherwise, if $i\in A\setminus C'$, then $i \not \in C$ because we have already seen that $C \subseteq C'$.
Since, by construction, $[v \, \bar{P}_j \, w \Leftrightarrow j \in C]$ for each $v, w \in \Omega$ with $v \leq^* \alpha <^* w$, we have that $f(\bar{R}'_{S\cup \{i\}}, \bar{R}_{-(S\cup \{i\})}) \, \bar{P}_i \, f(\bar{R}'_S, \bar{R}_{-S})$.
Agent $i$ then manipulates $f$ at $(\bar{R}'_{S}, \bar{R}_{-S})$ via $\bar{R}'_i$.
We can thus conclude that $f(\bar{R}'_A, \bar{R}_D) \leq \underline{\omega}(p(R))$. 

\medskip

\noindent Since $\omega(p(R')) >^* \alpha$ by assumption and $p(\bar{R}'_A, \bar{R}_D) = p(R')$, we have that $\omega(p(\bar{R}'_A, \bar{R}_D)) >^* \alpha$. 
If $\omega(p(R)) \in \Omega^2_C$ or if [$\omega(p(R)) \in \Omega$ and $\omega(p(R)) \not\in \omega(p(\bar{R}'_A, \bar{R}_D))$], then $f(\bar{R}'_A, \bar{R}_D) > \underline{\omega}(p(R))$.
This contradicts that $f(\bar{R}'_A, \bar{R}_D) \leq \underline{\omega}(p(R))$.  
We can thus conclude that $\omega(p(R)) \in \Omega$ and $\omega(p(R)) \in \omega(p(\bar{R}'_A, \bar{R}_D))$. 
Then, $\omega(p(R)) = \alpha$ and $\omega(p(\bar{R}'_A, \bar{R}_D)) = \{\alpha, \gamma\}$, with $\gamma > \alpha$. 
Since $f(\hat{R}') = \overline{\omega}(p(R'))$ and $p(\hat{R}')=p(R')=p(\bar{R}'_A, \bar{R}_D)$, we have that $f(\hat{R}') = \gamma$. 

\medskip

\noindent Next, we show that $f(\bar{R}_A, \hat{R}'_{D}) \geq \gamma$. 
Suppose by contradiction that $f(\bar{R}_A, \hat{R}'_{D}) < \gamma$. 
Starting at $\hat{R}'$, construct a sequence of profiles in which the preferences of all agents $i\in A$ are changed one by one from $\hat{R}'_i$ to $\bar{R}_i$ such that the sequence ends at $(\bar{R}_A, \hat{R}'_D)$. Since $f(\hat{R}')=\gamma> f(\bar{R}_A, \hat{R}'_D)$, there exists $S \subset A$ and $i \in (A \setminus S)$ such that $f(\bar{R}_S, \hat{R}'_{-S}) \geq \gamma$ and $f(\bar{R}_{S\cup \{i\}}, \hat{R}'_{-(S\cup \{i\})})< \gamma$.
Since, by Proposition \ref{omega}, $(\alpha, \gamma) \in \Omega^2_C$, we have $f(\bar{R}_{S\cup \{i\}}, \hat{R}'_{-(S\cup \{i\})}) \leq \alpha$. 

\medskip

\noindent If $i\in C'$, then, by construction,  $v \, \hat{P}'_i \, w$ for each $v, w \in \Omega$ with $v \leq \alpha < w$. 
Thus,  $f(\bar{R}_{S\cup\{i\}}, \hat{R}'_{-(S\cup\{i\})}) \, \hat{P}'_i \, f(\bar{R}_{S}, \hat{R}'_{-S})$.
Agent $i$ then manipulates $f$ at $(\bar{R}_{S}, \hat{R}'_{-S})$ via $\bar{R}_i$. 
Otherwise, if $i\in A\setminus C'$, then $i \not \in C$ because we have already seen that $C \subseteq C'$.
Since, by construction, $[v \, \bar{P}_j \, w \Leftrightarrow j \in C]$ for each $v, w \in \Omega$ with $v \leq^* \alpha <^* w$, we have that $f(\bar{R}_{S}, \hat{R}'_{-S}) \, \bar{P}_i \, f(\bar{R}_{S\cup\{i\}}, \hat{R}'_{-(S\cup\{i\})})$.
Agent $i$ then manipulates $f$ at $(\bar{R}_{S\cup\{i\}}, \hat{R}'_{-(S\cup\{i\})})$ via $\hat{R}'_i$. 
We can thus conclude that $f(\bar{R}_A, \hat{R}'_D) \geq \gamma$.

\medskip

\noindent Finally, it follows from $f(\bar{R}_A, \hat{R}'_D) \geq \gamma$ that $f(\bar{R}_A, \hat{R}'_D) > \alpha$. 
Since $\omega(p(R)) = \alpha$ and $p(\bar{R}_{A}, \hat{R}'_{D}) = p(R)$, we have that $f(\bar{R}_{A}, \hat{R}'_{D})=\alpha$, which is a contradiction. \end{proof}

\noindent To complete the proof of the proposition, we show that Definition \ref{voterscheme} is satisfied.

\medskip

\noindent $\Leftarrow$) By Lemma \ref{lemmamedian} ($i$) and ($ii$), we have that for each $R \in {\cal R}$ and each $\alpha \in r_\omega$, if $\{i \in A \, : \, p(R_i) \leq^* \alpha\} \in 
{\cal L}(\alpha) \setminus {\cal L}(\beta)$ for each $\beta <^* \alpha$, then $\omega(p(R))= \alpha$. 

\medskip

\noindent $\Rightarrow)$ Consider $\alpha \in r_\omega$, $\textbf{p} \in \Omega^a$ and $R \in {\cal R}$ such that $p(R) = \textbf{p}$ and $\omega(\textbf{p}) = \alpha$. 
We have to show that (a) $\{i \in A \, : \, p(R_i) \leq^* \alpha\} \in {\cal L}(\alpha)$ and (b) for each $\beta < \alpha$, $\{i \in A \, : \, p(R_i) \leq^* \beta\} \not\in {\cal L}(\beta)$. 
Since $\omega(\textbf{p}) = \alpha$ and $p(R) = \textbf{p}$, $\omega(p(R)) = \alpha$. 
Then, $\{i \in A \, : \, p(R_i) \leq^* \alpha\} \in {\cal L}(\alpha)$ by construction of ${\cal L}$. 
This establishes condition (a).
Next, suppose by contradiction that there exists $\beta <^* \alpha$ such that $\{i \in A \, : \, p(R_i) \leq^* \beta\} \in {\cal L}(\beta)$. 
Since we have already established that for each $\textbf{p}' \in \Omega^a$ and each $R' \in {\cal R}$ such that $p(R')=\textbf{p}'$, $\{i \in A: p(R'_i) \leq^* \gamma\} \in {\cal L}(\gamma)$ implies that $\omega(\textbf{p}') \leq^* \gamma$, it must be the case that $\{i \in A \, : \, p(R_i) \leq^* \beta\} \in {\cal L}(\beta)$ implies $\omega(\textbf{p}) \leq^* \beta$. 
This contradicts that $\beta <^* \alpha$.

\bigskip

\noindent Finally, we show that ${\cal L}$ complies with the conditions of Definition \ref{leftsystem}.
\begin{itemize}
\item By construction of ${\cal L}$, if $C \in {\cal L}(\alpha)$ for some $\alpha \in r_\omega$, then $C \in {\cal L}(\beta)$ for each $\beta >^* \alpha$. 
Therefore, condition ($ii$) of Definition \ref{leftsystem} is satisfied.  

\item We now show that ${\cal L}$ satisfies condition ($i$) of Definition \ref{leftsystem}. Consider any $\alpha \in r_{\omega}$, any $C \in {\cal L}(\alpha)$ and any $C' \supset C$, and suppose by contradiction that $C' \not\in {\cal L}(\alpha)$.

Suppose first that $\alpha \in \Omega$. Consider $R \in {\cal R}$ such that $p(R_i) \leq^* \alpha \Leftrightarrow i \in C$ and, for each $j \in D$, $d(R_j) = \alpha$. Since $C \in {\cal L}(\alpha)$, we have that by Definition~\ref{voterscheme}, $\omega(p(R)) \leq^* \alpha$  and, then, $f(R) \leq  \alpha$. Consider now a profile $\bar{R} = (R'_{C' \setminus C}, R_{N \setminus (C' \setminus C)}) \in {\cal R}$ such that, for each $i \in C' \setminus C$,  $p(R'_i) \leq^* \alpha$ and $x \, P'_i \, y$ for all $x \leq^* \alpha <^* y$. Then, $p(\bar{R}_i) \leq^* \alpha \Leftrightarrow i \in C'$. Since $C' \not\in {\cal L}(\alpha)$, we have that by Definition~\ref{voterscheme}, $\omega(p(\bar{R})) >^* \alpha$. If $\omega(p(\bar R)) \in \Omega^2_C$ and, since, by construction, $L(\bar{R}) = \emptyset$,\footnote{ 
Note that $\emptyset$ must not belong to $W(g_{p(\bar{R})})$ by the following arguments: ($i$) $W(g_{p(\bar{R})})$ should not contain all possible coalitions (as $\overline{\omega}(p(\bar{R})) \in \omega(p(\bar{R}))$); and ($ii$) Definition \ref{manjunath} implies that any superset of a coalition in $W(g_{p(\bar{R})})$ also belong to $W(g_{p(\bar R)})$.} then $g_{p(\bar{R})}(\bar{R}) = \overline{\omega}(p(\bar{R}))$ by Lemma~\ref{structure2}. Thus, $f(\bar R) > \alpha$ and the agent set $C' \setminus C$ manipulates $f$ at $\bar{R}$ via $R_{C' \setminus C}$. Hence, $f$ is not GSP. By Lemma \ref{equivalence}, $f$ is not SP.

Suppose now that $\alpha \in \Omega^2_C$. Consider $R \in {\cal R}$ such that $p(R_i) \leq^* \alpha \Leftrightarrow i \in C$ and for each $j \in D$, $d(R_j) >^* \alpha$. Since $C \in {\cal L}(\alpha)$, we have that by Definition \ref{voterscheme}, $\omega(p(R)) \leq^* \alpha$. By construction, $L(R) = D$,\footnote{Note that $D$ must belong to $W(g_{p(R)})$ by the following three arguments: ($i$) $W(g_{p(R)})$ should be non-empty (as $\underline{\omega}(p(R)) \in \omega(p(R))$); ($ii$) Definition \ref{manjunath} implies that any superset of a coalition in $W(g_{p(R)})$ also belong to $W(g_{p(R)})$; and ($iii$) $N(R)$ is empty by Proposition \ref{omega} and, then, the maximal possible value for $L(R)$ is $D$.} and therefore, $g_{p(R)}(R) = \underline{\omega}(p(R))$ by Lemma \ref{structure2}. Then, $f(R) <^* \alpha$. Consider now a profile $\bar{R} = (R'_{C' \setminus C}, R_{N \setminus (C' \setminus C)}) \in {\cal R}$ such that, for each $i \in C' \setminus C$,  $p(R'_i) \leq^* \alpha$ and $x \, P'_i \, y$ for all $x \leq^* \alpha <^* y$. Then, $p(\bar{R}_i) \leq^* \alpha \Leftrightarrow i \in C'$. Since $C' \not\in {\cal L}(\alpha)$, we have that by Definition \ref{voterscheme}, $\omega(p(\bar{R})) >^* \alpha$. Thus, $f(\bar{R}) >^* \alpha$, and the agent set $C' \setminus C$ manipulates $f$ at $\bar{R}$ via $R_{C' \setminus C}$. Hence, $f$ is not GSP. By Lemma \ref{equivalence}, $f$ is not SP.


\item With respect to condition ($iii$) of Definition \ref{leftsystem}, suppose that $\max\Omega$ exists and that $\max \Omega\notin r_\omega$. 
We show that for each $\alpha\in r_\omega\setminus\{\max r_\omega\}$, $\emptyset\in{\cal L}(\max r_\omega)\setminus {\cal L}(\alpha)$. 
We first show that for some $\alpha\in r_\omega$, $\emptyset\in{\cal L}(\alpha)$. 
Suppose by contradiction that for each $\alpha\in r_\omega$, $\emptyset\notin{\cal L}(\alpha)$. 
Consider $R\in{\cal R}$ such that for each $i\in A$, $p(R_i)=\max \Omega$. 
Then, for each $\alpha \in r_\omega$, $\{i\in A: p (R_i)\leq^{*}\alpha\}=\emptyset$. 
Since $\emptyset\notin{\cal L}(\alpha)$ for each $\alpha\in r_\omega$, we have that $\omega(p(R))\notin r_\omega$,
which is a contradiction. 

Second, we show that for each $\alpha\in r_\omega\setminus \{\max r_\omega\}$, $\emptyset\notin{\cal L}(\alpha)$. 
Suppose by contradiction that  for some $\alpha\in r_\omega\setminus\{ \max r_\omega\}$, $\emptyset\in{\cal L}(\alpha)$. Then, by condition ($i$) of Definition \ref{leftsystem}, we have that ${\cal L}(\alpha)=2^{A}$. 
Then, for each $R \in {\cal R}$, $\omega(p (R)) \leq^{*} \alpha$. 
Since $\max \Omega\notin r_\omega$, we have $\max r_\omega =\{y,\max \Omega\}$, with $y< \max \Omega$. 
Since $\alpha \in r_\omega\setminus\{ \max r_\omega\}$, we have that $\alpha \leq^{*} y$. 
Therefore, $\omega(p (R)) \leq^{*} y$ for each $R \in {\cal R}$.
This contradicts that $\max \Omega \in \Omega$. 

Therefore, condition ($iii$) of Definition \ref{leftsystem} is satisfied.  

\item Regarding condition ($iv$) of Definition \ref{leftsystem}, suppose that $\max\Omega$ does not exist and suppose by contradiction that for some $\alpha\in r_\omega$, $\emptyset\in{\cal L}(\alpha)$. 
Then, by condition ($i$) of Definition \ref{leftsystem}, ${\cal L}(\alpha)=2^{A}$. 
Therefore, $\omega(p(R))\leq^{*}\alpha$ for each $R\in{\cal R}$. Then, since $\mathbb{R} \setminus \Omega$ is either empty or the union of open sets, $\max r_\omega$ exists. 
Thus, $\max \Omega$ also exists, which is a contradiction.
Therefore, condition ($iv$) of Definition \ref{leftsystem} is satisfied.
\end{itemize}

\subsection*{Proof of Proposition \ref{second-step2}}

Given a generalized median voter function $\omega$ on a set $r_\omega$, with $r_{\omega}$ satisfying the conditions of Proposition \ref{omega}, 
we have to prove that $g_{\omega(\mathbf{p})}$ is a voting by collections of left-decisive sets for each $\omega(\mathbf{p}) \in r_{\omega}\cap\Omega^2_{C}$. 
That is, we have to show that for each $\omega(\mathbf{p})\in r_{\omega}\cap\Omega^2_{C}$, there is a minimal set of coalitions $W(g_{\omega(\mathbf{p})}) \in 2^N$ that satisfies the conditions in Definition \ref{winningdef}. 
First, define $W^{*}(g_{\omega(\mathbf{p})}) \subseteq 2^N$ as follows: $C \in W^{*}(g_{\omega(\mathbf{p})})$ if there is $R \in {\cal R}$ such that $\omega(p(R)) = \omega(\mathbf{p})$, $L_{\omega(\mathbf{p})}(R)=C$, and $g_{\omega(\mathbf{p})}(R)=\underline{\omega}(\mathbf{p})$ (and thus $f(R)=\underline{\omega}(\mathbf{p})$). 
By definition, for each $C\in W^{*}(g_{\omega(\mathbf{p})})$, $C\cap A\in {\cal L}(\omega(\mathbf{p}))\setminus {\cal L}(\underline{\omega}(\mathbf{p}))$. 
Next, let $W(g_{\omega(\mathbf{p})})$ be the set of the minimal coalitions of $W^{*}(g_{\omega(\mathbf{p})})$. Observe that $W(g_{\omega(\mathbf{p})})$ satisfies that for each $C \in W(g_{\omega(\mathbf{p})})$, $C \cap A \in {\cal L}(\omega(\mathbf{p}))\setminus {\cal L}(\underline{\omega}(\mathbf{p}))$.

\bigskip

\noindent \textit{Step 1: We show that if $C \in W(g_{\omega(\mathbf{p})})$, then for each $R' \in {\cal R}$ such that $\omega(p(R')) = \omega(\mathbf{p})$ and $C\subseteq L_{\omega(\mathbf{p})}(R')$, we have that $g_{\omega(\mathbf{p})}(R')=\underline{\omega}(\mathbf{p})$.}

\medskip

\noindent Suppose by contradiction that $C \in W(g_{\omega(\mathbf{p})})$ but there is $\bar{R} \in {\cal R}$ such that  $\omega(p(\bar{R}))=\omega(\mathbf{p})$, $C\subseteq L_{\omega(\mathbf{p})}(\bar{R})$, and $g_{\omega(\mathbf{p})}(\bar{R})=\overline{\omega}(\mathbf{p})$ (and thus $f(\bar{R})=\overline{\omega}(\mathbf{p})$). 
Since $C\in W(g_{\omega(\mathbf{p})})$, there is $R \in {\cal R}$ such that $\omega(p(R)) = \omega(\mathbf{p})$, $L_{\omega(\mathbf{p})}(R)=C$, and $g_{\omega(\mathbf{p})}(R)=\underline{\omega}(\mathbf{p})$ (and thus $f(R)=\underline{\omega}(\mathbf{p})$). 

\medskip

\noindent First, suppose that $L_{\omega(\mathbf{p})}(\bar{R})=C$ and consider $(\bar{R}_C, R_{-C})\in{\cal R}$. 
Since $\{i\in A: p((\bar{R}_C, R_{-C})_i)\leq^{*}\omega(\mathbf{p})\}=\{i\in A: p(R_i)\leq^{*}\omega(\mathbf{p})\}\in{\cal L}(\omega(\mathbf{p}))\setminus {\cal L}(\underline{\omega}(\mathbf{p}))$, it follows from Proposition \ref{generalized} that $\omega(p(\bar{R}_C, R_{-C}))=\omega(\mathbf{p})$. 
If $g_{\omega(\mathbf{p})}(\bar{R}_C, R_{-C})=\underline{\omega}(\mathbf{p})$, then $f(\bar{R}_C, R_{-C})=\underline{\omega}(\mathbf{p})$ and the agent set $N\setminus C$ manipulates $f$ at this profile via $\bar{R}_{-C}$. 
Otherwise, if  $g_{\omega(\mathbf{p})}(\bar{R}_C, R_{-C})=\overline{\omega}(\mathbf{p})$, then $f(\bar{R}_C, R_{-C})=\overline{\omega}(\mathbf{p})$ and the agent set $C$ manipulates $f$ at this profile via $R_C$. In both cases, $f$ is not GSP. By Lemma \ref{equivalence}, $f$ is not SP.
Hence, we have shown that for each $R' \in {\cal R}$ such that $\omega(p(R'))=\omega(\mathbf{p})$ and $L_{\omega(\mathbf{p})}(R') = C$, we have that $g_{\omega(\mathbf{p})}(R')=\underline{\omega}(\mathbf{p})$.

\medskip

\noindent Second, suppose that $C\subset L_{\omega(\mathbf{p})}(\bar{R}) \equiv B$ and consider $(\bar{R}_C, R_{-C})\in{\cal R}$. 
Since $\{i\in A: p((\bar{R}_C, R_{-C})_i)\leq^{*}\omega(\mathbf{p})\}=\{i\in A: p(R_i)\leq^{*}\omega(\mathbf{p})\}\in{\cal L}(\omega(\mathbf{p}))\setminus {\cal L}(\underline{\omega}(\mathbf{p}))$ it follows from Proposition \ref{generalized} that $\omega(p(\bar{R}_C, R_{-C}))=\omega(\mathbf{p})$. 
Given that $L_{\omega(\mathbf{p})}(\bar{R}_C, R_{-C})=C$, we have, by the previous paragraph, that $g_{\omega(\mathbf{p})}(\bar{R}_C, R_{-C})=\underline{\omega}(\mathbf{p})$.
Thus, $f(\bar{R}_C, R_{-C})=\underline{\omega}(\mathbf{p})$. 
Consider now $(\bar{R}_{B}, R_{-B})\in {\cal R}$. 
Since $\{i\in A: p((\bar{R}_{B},R_{-B})_i)\leq^{*}\omega(\mathbf{p})\}=\{i\in A: p(\bar{R}_i)\leq^{*}\omega(\mathbf{p})\}\in{\cal L}(\omega(\mathbf{p}))\setminus {\cal L}(\underline{\omega}(\mathbf{p}))$, it follows from Proposition \ref{generalized} that $\omega(p(\bar{R}_{B}, R_{-B}))= \omega(\mathbf{p})$. 
If $g_{\omega(\mathbf{p})}(\bar{R}_{B},R_{-B})=\underline{\omega}(\mathbf{p})$, then $f(\bar{R}_{B},R_{-B})=\underline{\omega}(\mathbf{p})$ and the agent set $N\setminus B$ manipulates $f$ at this profile via $\bar{R}_{N\setminus B}$. 
Otherwise, if $g_{\omega(\mathbf{p})}(\bar{R}_{B},R_{-B})=\overline{\omega}(\mathbf{p})$, then $f(\bar{R}_{B},R_{-B})=\overline{\omega}(\mathbf{p})$ and the agent set $B\setminus C$ manipulates $f$ at this profile via $R_{B \setminus C}$. 
In both cases, $f$ is not GSP. By Lemma \ref{equivalence}, $f$ is not SP.

\bigskip

\noindent \textit{Step 2: We show that for each $C\in W(g_{(\omega(\mathbf{p})})$, $C\cap D\neq \emptyset$.}

\medskip

\noindent Suppose by contradiction that there is $C\in W(g_{\omega(\mathbf{p})})$ such that $C\cap D=\emptyset$. 
Then, $C\cap A=C\in {\cal L}(\omega(\mathbf{p})) \setminus {\cal L}(\underline{\omega}(\mathbf{p}))$. 
Consider $\mathbf{p}' \in \Omega^{a}$ such that for each $i\in A$, $\mathbf{p}'_{i}\leq^{*}\omega(\mathbf{p})$ if and only if $i\in C$. 
Since $C\in {\cal L}(\omega(\mathbf{p})) \setminus {\cal L}(\underline{\omega}(\mathbf{p}))$, $\omega(\mathbf{p}')=\omega(\mathbf{p})$.
Thus, there is $R\in{\cal R}$ such that $p(R)=\mathbf{p}'$ and $f(R)=\overline{\omega}(\mathbf{p})$. 
Observe that $C\subseteq L_{\omega(\mathbf{p})}(R)$.
Therefore, by Step 1, we have that $g_{\omega(\mathbf{p})}(R)=\underline{\omega}(\mathbf{p})$.
Thus, $f(R)=\underline{\omega}(\mathbf{p})$, which is a contradiction.

\bigskip

\noindent \textit{Step 3: We show that for each minimal coalition $B$ of ${\cal L}(\omega(\mathbf{p}))\setminus{\cal L}(\underline{\omega}(\mathbf{p}))$, there is $C\in W(g_{\omega(\mathbf{p})})$ such that $C\cap A=B$.}

\medskip

\noindent Suppose by contradiction that for some minimal coalition $B$ of ${\cal L}(\omega(\mathbf{p}))\setminus{\cal L}(\underline{\omega}(\mathbf{p}))$, there is no $C\in W(g_{\omega(\mathbf{p})})$ such that $C\cap A=B$. 
Consider $\mathbf{p}' \in \Omega^{a}$ such that for each $i\in A$, $\mathbf{p}'_{i}\leq^{*}\omega(\mathbf{p})$ if and only if $i\in B$. Since $B\in{\cal L}(\omega(\mathbf{p}))\setminus{\cal L}(\underline{\omega}(\mathbf{p}))$, $\omega(\mathbf{p}')=\omega(\mathbf{p})$.
Thus, there is $R\in{\cal R}$ such that $p(R)=\mathbf{p}'$ and $f(R)=\underline{\omega}(\mathbf{p})$. 
Observe that $L_{\omega(\mathbf{p})}(R)\cap A=B$. 
Since $B$ is minimal in ${\cal L}(\omega(\mathbf{p}))\setminus{\cal L}(\underline{\omega}(\mathbf{p}))$ and for each $C\in W(g_{\omega(\mathbf{p})})$, $C\cap A\neq B$, we have that there is no $C\in W(g_{\omega(\mathbf{p})})$ such that $C \subseteq L_{\omega(\mathbf{p})}(R)$. 
Therefore, $g_{\omega(\mathbf{p})}(R)=\overline{\omega}(\mathbf{p})$.
Thus, $f(R)=\overline{\omega}(\mathbf{p})$, which is a contradiction. 

\subsection*{Proof of Theorem \ref{theorem}}

\noindent The equivalence between $(i)$ and $(ii)$ is due to Lemma \ref{equivalence}. 
We complete the proof by showing the equivalence between $(i)$ and $(iii)$.
If $(i)$ holds, then the structure of $f$ is as described in $(iii)$, given Propositions \ref{generalized} and \ref{second-step2}. 
To show that $(iii)$ implies $(i)$, consider any $f$ that is decomposable as described in $(iii)$ with a generalized median voter function $\omega$ on $r_\omega$, with $r_{\omega}$ satisfying the conditions of Proposition \ref{omega}, 
and a set of voting by collections of left-decisive sets $\{g_{\omega(\mathbf{p})} \, : \, {\cal R} \rightarrow \omega(\mathbf{p})\}_{\omega(\mathbf{p})\in r_{\omega}\cap\Omega^2_{C}}$. 
Suppose by contradiction that $f$ is not SP. 
Then, there is a profile $R\in{\cal R}$ and an agent $i\in N$ with the alternative preference $R'_i\in{\cal R}_{i}$ such that $f(R'_i,R_{-i}) \, P_i \, f(R)$. 
We assume without loss of generality that $f(R)<f(R'_i,R_{-i})$.

\medskip

\noindent Suppose first that $\omega(p(R))=\omega(p(R'_i,R_{-i}))$. 
If $\omega(p(R)) \in \Omega$, then $f(R)=f(R'_i, R_{-i})$, which contradicts $f(R'_i,R_{-i}) \, P_i \, f(R)$. 
If $\omega(p(R)) \in \Omega^2_C$, then $f(R)=\underline{\omega}(p(R))$ and $f(R'_i,R_{-i})=\overline{\omega}(p(R))$. Since $f(R)=\underline{\omega}(p(R))$, we know from the definition of the voting by collections of left-decisive sets that $g_{\omega(p(R))}(R)=\underline{\omega}(\mathbf{p})$. 
Therefore, $C\subseteq L_{\omega(p(R))}(R)$ for some $C\in W(g_{\omega(p(R))})$.
Since $f(R'_i,R_{-i}) \, P_i \, f(R)$, $i\notin  L_{\omega(p(R))}(R)$. 
Observe then that $L_{\omega(p(R))}(R)\subseteq L_{\omega(p(R))}(R'_i,R_{-i})$.
Thus, $C \subseteq L_{\omega(p(R))}(R'_i,R_{-i})$. 
Therefore, $g_{\omega(p(R))}(R'_i,R_{-i})=\underline{\omega}(p(R))$ and, consequently, $f(R'_i,R_{-i})=\underline{\omega}(p(R))$.
This is a contradiction. 

\medskip

\noindent Suppose now that $\omega(p(R))\neq\omega(p(R'_i,R_{-i}))$. 
Then, $i\in A$. 
Given that $f(R)<f(R'_i,R_{-i})$, we have that $\omega(p(R))<^{*}\omega(p(R'_i,R_{-i}))$. 
Since $f(R'_i,R_{-i}) \, P_i \, f(R)$, we have that if $\omega(p(R)) \in \Omega$, then $p(R_i)> \omega(p(R))$.
Also, if $\omega(p(R)) \in \Omega^2_C$, then $p(R_i)\geq \overline{\omega}(p(R))$. 
Thus, in any case, $p(R_i)>^{*}\omega(p(R))$. 
Then, $i\notin\{j\in A: p(R_j)\leq^{*} \omega(p(R))\}$. 
Observe that $\{j\in A: p(R_j)\leq^{*} \omega(p(R))\}\subseteq \{j\in A: p((R'_i,R_{-i})_{j})\leq^{*} \omega(p(R))\}$. 
Therefore, by the definition of generalized median voter function, $\omega(p(R'_i,R_{-i}))\leq^*\omega(p(R))$.
This is a contradiction. 

\subsection*{Proof of Proposition \ref{PE}}

First, we prove $(i)$. 
Consider any SP rule $f$ such that $\Omega=X$. 
Suppose by contradiction that $f$ is not PE. 
Then, there is $x\in X$ and $R\in{\cal R}$ such that $x\, P_i \, f(R)$ for each $i\in N$. 
Since $\Omega=X$, $x\in \Omega$.
Thus, there is a profile $R'\in{\cal R}$ such that $f(R')=x$. 
Then, the agent set $N$ manipulates $f$ at $R$ via $R'$.
Hence, $f$ is not GSP. 
By Lemma \ref{equivalence}, $f$ is not SP.

\medskip

\noindent Second, we prove $(ii)$. 
Consider any SP rule $f$ such that $\Omega\notin\{X,\{\min X, \max X\}\}$. 
Therefore, $X\setminus \Omega\neq \emptyset$. 
Suppose by contradiction that $f$ is PE. 
We consider two subcases.
\begin{itemize}
\item Suppose that $\min X$ and $\max X$ exist, but $\{\min X,\max X\}\not\subset \Omega$. 
Assume without loss of generality that $\min X \notin \Omega$. 
Consider $R\in{\cal R}$ such that for each $i \in A$, $\rho(R_i) = \min X$ and for each $i \in D$, $\delta(R_i)=\max X$. By construction, $\min X \, P_i \, x$ for each $x \in \Omega$ and each $i \in N$. 
Then, $\min X$ Pareto dominates $f(R)$ because $f(R)\in \Omega $. 
\item Suppose that either [$\min X$ or $\max X$ does not exist] or that [both exist but $\{\min X,$ $\max X\} \subsetneq \Omega$]. Let $x \in X \setminus \Omega$ and let $y \in \mathbf{int}(\Omega)$ be the closest alternative to $x$ in $\mathbf{int}(\Omega)$ (if there are two equally closest alternatives, consider any of them).
Hence, $(x, y) \cap \Omega = \emptyset$ whenever $x < y$ and $(y, x) \cap \Omega = \emptyset$ whenever $x > y$. Observe that $x$ and $y$ always exist because $\Omega\notin\{X,\{\min X, \max X\}\}$ and $\mathbb{R} \setminus \Omega$ is either empty or the union of open sets. 
Consider $R'\in{\cal R}$ such that for each $i\in A$, $[\rho(R'_i)=x$ and $p(R'_i)=y]$, and for each $i\in D$, $\delta(R'_i)=y$. 
Then, $\{j\in A: p(R'_j)\leq^{*}y\}=A$ and for each $\alpha<^{*}y$, $\{j\in A: p(R'_j)\leq^{*}\alpha\}=\emptyset$. Therefore, by Theorem \ref{theorem}, $\omega(p(R'))=y$ and $f(R')=y$. 
A contradiction is reached by observing that $x \, P'_i \,y $ for each $i\in N$.
\end{itemize}

\noindent Finally, we prove $(iii)$. Suppose that $\min X$ and $\max X$ exist and consider any SP rule $f$ such that $\Omega=\{\min X, \max X\}\neq X$. We divide the proof into three steps:

\bigskip

\noindent\textit{Step 1: We show that if $A=\emptyset$, then $f$ is PE.}

\medskip

\noindent Since $A=\emptyset$, $N=D\neq \emptyset$. 
Then, for each $R \in {\cal R}$, $\omega(p(R)) = (\min X, \max X)$. 
Observe that for each agent, either $\min X$ or $\max X$ is the most preferred alternative of $X$. 
Consider first $R\in{\cal R}$ such that all agents have the same most preferred alternative of $X$. Assume, without loss of generality, that this most preferred alternative is $\min X$. 
Then, $L_{(\min X, \max X)}(R)=N$ and, by Theorem \ref{theorem}, $g_{(\min X, \max X)}(R)= \min X$ and $f(R)=\min X$. 
Consider now $R'\in{\cal R}$ such that not all agents have the same most preferred alternative. 
Then, there is a non-empty set of agents $S\subset D$ whose most preferred alternative is $\min X$, while for the remaining agents, $N\setminus S$, $\max X$ is the most preferred alternative. 
Since $\Omega=\{\min X, \max X\}$, we have $f(R')\in\{\min X, \max X\}$. 
If, on the one hand, $f(R')=\min X$, then $f(R') \, P'_i \, x$ for each $i\in S$ and each $x\in X\setminus \{\min X\}$. 
If, on the other hand, $f(R')=\max X$, then $f(R') \, P'_{i} \, x$ for each $i\in N\setminus S$ and each $x\in X\setminus \{\max X\}$. 
Hence, $f$ is PE.

\bigskip

\noindent\textit{Step 2: We show that if $A\neq\emptyset$ and $r_\omega = \{(\min X, \max X)\}$, then $f$ is PE.}

\medskip

\noindent Since $\max\Omega=\max X$ exists and $\max \Omega \notin r_\omega$, we know, by Theorem \ref{theorem} (exactly by condition $(iii)$ of Definition \ref{leftsystem}) that $\emptyset\in {\cal L}(\min X, \max X)$. 
Then, $\emptyset$ is a minimal coalition of ${\cal L}(\min X, \max X)$.
By Theorem \ref{theorem} (exactly by condition $(ii)$ of Definition \ref{winningdef}), there is $C\in W(g_{(\min X, \max X)})$ such that $C \cap A = \emptyset$.
Thus, $C\subseteq D$ and $C \neq \emptyset$.
Note that for each agent in $D$, either $\min X$ or $\max X$ is the most preferred alternative of $X$. 
Observe also that since $r_\omega = \{(\min X, \max X)\}$, for each $R\in{\cal R}$, $\omega(p(R))=(\min X,\max X)$. Consider first $R'\in{\cal R}$ such that $C\subseteq L_{(\min X,\max X)}(R')$. 
Then, by Theorem \ref{theorem}, $g_{(\min X,\max X)}(R')=\min X$ and $f(R')=\min X$. 
Since $C\subseteq L_{(\min X,\max X)}(R')$ and $C\subseteq D$, we have that $\min X \, P'_{i} \, x$ for each $i\in C$ and each $x\in X\setminus\{\min X\}$.
Consider now $R''\in{\cal R}$ such that $C\not\subseteq L_{(\min X,\max X)}(R'')$. 
Then, by Theorem \ref{theorem}, $g_{(\min X,\max X)}(R'')=\max X$ and $f(R'')=\max X$. 
Since $C\not\subseteq L_{(\min X,\max X)}(R'')$ and $C\subseteq D$, we have that $\max X \, P''_{i} \, x$ for each $i\in C \setminus L_{(\min X,\max X)}(R'')$ and each $x\in X\setminus\{\max X\}$. 
Hence, $f$ is PE.

\bigskip

\noindent\textit{Step 3: We show that if $A\neq\emptyset$ and $r_\omega \neq \{(\min X,\max X)\}$, then $f$ is not PE.}\medskip

\noindent Since $\Omega=\{\min X, \max X\}$ and $r_\omega \neq \{(\min X,\max X)\}$, we have that $\min X\in r_\omega$ and/or $\max X\in r_\omega$. 
Suppose without loss of generality that $\min X \in r_\omega$. 
Let $x\not\in \Omega$ and consider $R\in {\cal R}$ such that for each $i\in A$, $[\rho(R_i)=x$ and $p(R_i)=\min X]$, and for each $i\in D$, $\delta(R_i)=\min X$. 
Since $\min X\in r_\omega$ and $\{j\in A: p(R_j)\leq \min X\}=A$, it follows from Theorem \ref{theorem} that $\omega(p(R))=\min X$ and $f(R)=\min X$. However, $x \, P_i \, \min X$ for each $i\in N$.
Thus, $f$ is not PE.

\bibliographystyle{ecta}
\bibliography{socialchoice}

\begin{thebibliography}{22}
\newcommand{\enquote}[1]{``#1''}
\expandafter\ifx\csname natexlab\endcsname\relax\def\natexlab#1{#1}\fi

\bibitem[\protect\citeauthoryear{Achuthankutty and Roy}{Achuthankutty and
  Roy}{2018}]{achuthankutty2018dictatorship}
\textsc{Achuthankutty, G. and S.~Roy} (2018): \enquote{Dictatorship on
  top-circular domains,} \emph{Theory and Decision}, 85, 479--493.

\bibitem[\protect\citeauthoryear{Alcalde-Unzu and Vorsatz}{Alcalde-Unzu and
  Vorsatz}{2018}]{alcalde2018strategy}
\textsc{Alcalde-Unzu, J. and M.~Vorsatz} (2018): \enquote{Strategy-proof
  location of public facilities,} \emph{Games and Economic Behavior}, 112,
  21--48.

\bibitem[\protect\citeauthoryear{Alcalde-Unzu and Vorsatz}{Alcalde-Unzu and
  Vorsatz}{2023}]{alcalde2023structure}
---\hspace{-.1pt}---\hspace{-.1pt}--- (2023): \enquote{The structure of
  strategy-proof rules,} \emph{arXiv:2304.12843}.

\bibitem[\protect\citeauthoryear{Barber{\`a}}{Barber{\`a}}{2011}]{barbera2011strategyproof}
\textsc{Barber{\`a}, S.} (2011): \enquote{Strategyproof social choice,}
  \emph{Handbook of Social Choice and Welfare}, 2, 731--831.

\bibitem[\protect\citeauthoryear{Barber{\`a}, Berga, and Moreno}{Barber{\`a}
  et~al.}{2010}]{barbera2010individual}
\textsc{Barber{\`a}, S., D.~Berga, and B.~Moreno} (2010): \enquote{Individual
  versus group strategy-proofness: When do they coincide?} \emph{Journal of
  Economic Theory}, 145, 1648--1674.

\bibitem[\protect\citeauthoryear{Barber{\`a}, Berga, and Moreno}{Barber{\`a}
  et~al.}{2012}]{barbera2012domains}
---\hspace{-.1pt}---\hspace{-.1pt}--- (2012): \enquote{Domains, ranges and
  strategy-proofness: the case of single-dipped preferences,} \emph{Social
  Choice and Welfare}, 39, 335--352.

\bibitem[\protect\citeauthoryear{Barber{\`a}, Berga, and Moreno}{Barber{\`a}
  et~al.}{2022}]{barbera2022restricted}
---\hspace{-.1pt}---\hspace{-.1pt}--- (2022): \enquote{Restricted environments
  and incentive compatibility in interdependent values models,} \emph{Games and
  Economic Behavior}, 131, 1--28.

\bibitem[\protect\citeauthoryear{Barber{\`a} and Jackson}{Barber{\`a} and
  Jackson}{1994}]{barbera1994characterization}
\textsc{Barber{\`a}, S. and M.~Jackson} (1994): \enquote{A characterization of
  strategy-proof social choice functions for economies with pure public goods,}
  \emph{Social Choice and Welfare}, 11, 241--252.

\bibitem[\protect\citeauthoryear{Berga}{Berga}{1998}]{berga1998strategy}
\textsc{Berga, D.} (1998): \enquote{Strategy-proofness and single-plateaued
  preferences,} \emph{Mathematical Social Sciences}, 35, 105--120.

\bibitem[\protect\citeauthoryear{Berga and Serizawa}{Berga and
  Serizawa}{2000}]{berga2000maximal}
\textsc{Berga, D. and S.~Serizawa} (2000): \enquote{Maximal domain for
  strategy-proof rules with one public good,} \emph{Journal of Economic
  Theory}, 90, 39--61.

\bibitem[\protect\citeauthoryear{Black}{Black}{1948{\natexlab{a}}}]{black1948decisions}
\textsc{Black, D.} (1948{\natexlab{a}}): \enquote{The decisions of a committee
  using a special majority,} \emph{Econometrica}, 245--261.

\bibitem[\protect\citeauthoryear{Black}{Black}{1948{\natexlab{b}}}]{black1948rationale}
---\hspace{-.1pt}---\hspace{-.1pt}--- (1948{\natexlab{b}}): \enquote{On the
  rationale of group decision-making,} \emph{Journal of Political Economy}, 56,
  23--34.

\bibitem[\protect\citeauthoryear{Bossert and Peters}{Bossert and
  Peters}{2014}]{bossert2014single}
\textsc{Bossert, W. and H.~Peters} (2014): \enquote{Single-basined choice,}
  \emph{Journal of Mathematical Economics}, 52, 162--168.

\bibitem[\protect\citeauthoryear{Feigenbaum and Sethuraman}{Feigenbaum and
  Sethuraman}{2015}]{feigenbaum2015strategyproof}
\textsc{Feigenbaum, I. and J.~Sethuraman} (2015): \enquote{Strategyproof
  mechanisms for one-dimensional hybrid and obnoxious facility location
  models,} in \emph{Workshops at the twenty-ninth AAAI conference on artificial
  intelligence}.

\bibitem[\protect\citeauthoryear{Gibbard}{Gibbard}{1973}]{gibbard1973manipulation}
\textsc{Gibbard, A.} (1973): \enquote{Manipulation of voting schemes: a general
  result,} \emph{Econometrica: journal of the Econometric Society}, 587--601.

\bibitem[\protect\citeauthoryear{Jennings, Laraki, Puppe, and Varloot}{Jennings
  et~al.}{2023}]{jennings2023new}
\textsc{Jennings, A.~B., R.~Laraki, C.~Puppe, and E.~M. Varloot} (2023):
  \enquote{New characterizations of strategy-proofness under
  single-peakedness,} \emph{Mathematical Programming}, 1--32.

\bibitem[\protect\citeauthoryear{Manjunath}{Manjunath}{2014}]{manjunath2014efficient}
\textsc{Manjunath, V.} (2014): \enquote{Efficient and strategy-proof social
  choice when preferences are single-dipped,} \emph{International Journal of
  Game Theory}, 43, 579--597.

\bibitem[\protect\citeauthoryear{Moulin}{Moulin}{1980}]{moulin1980strategy}
\textsc{Moulin, H.} (1980): \enquote{On strategy-proofness and single
  peakedness,} \emph{Public Choice}, 35, 437--455.

\bibitem[\protect\citeauthoryear{Moulin}{Moulin}{1984}]{moulin1984generalized}
---\hspace{-.1pt}---\hspace{-.1pt}--- (1984): \enquote{Generalized
  Condorcet-winners for single peaked and single-plateau preferences,}
  \emph{Social Choice and Welfare}, 1, 127--147.

\bibitem[\protect\citeauthoryear{Rodr{\'\i}guez-{\'A}lvarez}{Rodr{\'\i}guez-{\'A}lvarez}{2017}]{rodriguez2017single}
\textsc{Rodr{\'\i}guez-{\'A}lvarez, C.} (2017): \enquote{On single-peakedness
  and strategy-proofness: ties between adjacent alternatives,} \emph{Economics
  Bulletin}, 37, 1966--1974.

\bibitem[\protect\citeauthoryear{Satterthwaite}{Satterthwaite}{1975}]{satterthwaite1975strategy}
\textsc{Satterthwaite, M.~A.} (1975): \enquote{Strategy-proofness and Arrow's
  conditions: Existence and correspondence theorems for voting procedures and
  social welfare functions,} \emph{Journal of Economic Theory}, 10, 187--217.

\bibitem[\protect\citeauthoryear{Thomson}{Thomson}{2022}]{thomson2022should}
\textsc{Thomson, W.} (2022): \enquote{Where should your daughter go to college?
  An axiomatic analysis,} \emph{Social Choice and Welfare}, 1--18.

\end{thebibliography}

\end{document}